\documentclass[12pt,a4]{article}

\title{Monetizing digital content with network effects: A mechanism-design approach\thanks{For useful comments and suggestions, we thank  audiences at BSoE, UC3M, Lisbon GT Meeting'23, CMID'24, EARIE'24, VfS'24, 
and, in particular, 
Andreas Asseyer, Helmut Bester, Tobit Gamp, \'Angel Hernando-Veciana, Matthias Lang, Volker Nocke, Johannes Schneider, Sebastian Schweighofer-Kodritsch, and Roland Strausz.
Financial support by the European Union through the ERC-grant PRIVDIMA (project number 101096682) and by Deutsche Forschungsgemeinschaft through CRC TRR 190 is gratefully acknowledged.}}
\author{Vincent Meisner\thanks{Humboldt-Universit\"at zu Berlin, Institute for Economic Theory 1, Spandauer Str. 1, D-10178 Berlin (Germany), Email: vincent.meisner@tu-berlin.de.} \and  Pascal Pillath\thanks{Humboldt-Universit\"at zu Berlin, Institute for Economic Theory 1, Spandauer Str. 1, D-10178 Berlin (Germany), Email: pascal.pillath@hu-berlin.de.}}

\date{\today}

\usepackage{xcolor}
\usepackage{colortbl}
\usepackage{booktabs}
\usepackage{hhline}
\usepackage{amsmath}
\usepackage{array}
\usepackage{amssymb}
\usepackage{bbm}
\usepackage{amsthm}
 \usepackage{arcs}
\usepackage[latin1]{inputenc}
\usepackage{parskip}
\usepackage{sgame}
\usepackage{color}
\usepackage{cancel}
\usepackage{hyperref}
\usepackage{wrapfig}
\usepackage{graphics}
\usepackage{graphicx}
\usepackage{tikz,pgfplots}
\usepackage{verbatim}
\usepackage[round]{natbib}
\usepackage{subfig}
\usepackage{gensymb}
\usepackage{pdfpages}
\usepackage[paper=a4paper,margin=1.25in]{geometry}
\usepackage{thmtools, thm-restate}
\usepackage{cancel}
\usepackage{titlesec}
\usepackage{enumerate}
\usepackage{nicefrac}
\usepackage{adjustbox}
\usepackage{pdflscape}
\usepackage{outlines}
\usepackage{multirow}

\usepackage{algpseudocode}
\usepackage{algorithm}

\usepackage{tikz}
\usetikzlibrary{tikzmark}

\DeclareMathOperator*{\argmax}{arg\,max}

\newtheoremstyle{slplain}
  {.5\baselineskip\@plus.2\baselineskip\@minus.2\baselineskip}
  {.5\baselineskip\@plus.2\baselineskip\@minus.2\baselineskip}
  {\slshape}
  {}
  {\bfseries}
  {.}
  { }
  {}

\DeclareMathOperator{\sign}{sign}

\theoremstyle{slplain}
\newtheorem{prop}{Proposition}
\newtheorem{lem}{Lemma}

\newcommand{\ea}[1]{\begin{align*}#1\end{align*}}
\newcommand{\eq}[1]{\begin{equation}#1\end{equation}}
\newcommand{\ean}[1]{\begin{align}#1\end{align}}

\usetikzlibrary{positioning,arrows,patterns,decorations.pathreplacing}
\tikzset{
  state/.style={circle,draw,minimum size=6ex},
  arrow/.style={-latex, shorten >=1ex, shorten <=1ex}}

\DeclareFontFamily{U}{mathx}{\hyphenchar\font45}
\DeclareFontShape{U}{mathx}{m}{n}{
      <5> <6> <7> <8> <9> <10>
      <10.95> <12> <14.4> <17.28> <20.74> <24.88>
      mathx10
      }{}
\DeclareSymbolFont{mathx}{U}{mathx}{m}{n}
\DeclareFontSubstitution{U}{mathx}{m}{n}
\DeclareMathAccent{\widecheck}{0}{mathx}{"71}
\DeclareMathAccent{\wideparen}{0}{mathx}{"75}

\pgfplotsset{soldot/.style={color=black,only marks,mark=*}} \pgfplotsset{holdot/.style={color=black,fill=white,only marks,mark=*}}

\hypersetup{nolinks=true}

\begin{document}

\maketitle
\begin{abstract} \noindent
We design profit-maximizing mechanisms to sell an excludable and non-rival good with positive and/or negative network effects. Buyers have heterogeneous private values that depend on how many others also consume the good. In optimum, an endogenous number of the highest types consume the good, and we can implement this allocation in dominant strategies. We apply our insights to digital content creation, and we are able to rationalize features seen in monetization schemes in this industry such as voluntary contributions, community subsidies, and exclusivity bids.
\end{abstract}
\noindent 
JEL-Classification: D82.

\newpage


\section{Introduction} \label{sec:intro}

The ``creator economy" is a system in which creators provide digital content to users. 
With creators on, for instance, Instagram, OnlyFans, Snapchat, Tiktok, Twitch, or YouTube, the size of this global market is estimated to be over 250 billion US dollars in 2023.\footnote{See, e.g., https://www.goldmansachs.com/insights/articles/the-creator-economy-could-approach-half-a-trillion-dollars-by-2027.}
While this industry has many examples of top earners, the vast majority of creators cater to smaller audiences and struggle to make a living from their content production.\footnote{Some of the ``Top Creators 2023" listed by \cite{forbes2023} such as MrBeast, FuckJerry, Jake Paul, or KSI made more than \$30 million dollars that year, whereas The \cite{economist2021} suggests that more than 99\% of content creators barely earn below minimum wage.} 
Nevertheless, this career choice is becoming increasingly popular \citep[see, e.g.,][]{aridor2024}, and 
various payment features to support smaller creators have arisen. 
In this paper, we contribute to the design of optimal monetization schemes when audiences are small. Our model emphasizes the implications from three defining features of this market: oligopsony, network effects, and the non-rivalry of digital content. We show that the optimal payment structure is fundamentally different from the simple pricing that maximizes profit when audiences are large, and externalities vanish.

We model the content creator as a mechanism designer selling a non-rival and excludable good (her digital content) that she produces at a fixed cost.\footnote{Rather than literally taking this cost as a cost for production, we can also more broadly interpret it as an opportunity cost arising from not working in a different industry.} Buyers (users) draw a private value type, and their full consumption value is determined by a given function that also depends on how many other buyers access the good. 
Additionally, there is a direct network effect on the seller's profit.
Because both network effects only depend on the number of consumers and not their identity, always the highest value types are selected for consumption, but the optimal number of selected consumers is endogenous. 
This allocation is implementable in dominant strategies, and the solution to an optimization under a weaker Bayesian incentive constraint is identical. 
We can express the direction and size of network effects through a single parameter, which depending on its level can be interpreted as a degree of rivalry or a degree of excludability of the good. With this parsimonious approach, we nest important benchmark cases: selling an indivisible private good \citep{myerson1981}, which is essentially a good that loses its value when it is shared, and the private supply of a public good \citep{guth1986private}, which can be seen as a good that loses its value when it is not shared with everyone.

We first show that the profit-maximizing allocation can be implemented with a voluntary all-pay contribution mechanism in which users can opt to pay more than others. 
The optimality of such schemes is an implication of digital content being a club good with a fixed production cost. Here, the cost entails a positive externality among the buyers independent of the direction of the network effects because another buyer might be necessary to finance the production. To increase the probability that the content is provided at all, a high-value consumer is willing to pay more than others for an identical good. A wide range of platforms such as OnlyFans, Substack, or Twitch couple subscription fees with ``tipping" or ``donation" features. This aspect loses importance when markets become large such that a simple posted price mechanism becomes optimal. 

While the argument above does not hinge on network effects, we also illustrate how positive network effects offer a self-interested rationale for subsidies among users.
For example, Twitch employs a feature allowing users to gift subscriptions to others. While this may seem like an altruistic element, consumers in our model pay to support the growth of a creator's community because its size enters their consumption value. For instance, users may value engagement with others through comments, ``likes," and chats.
We also allow for value network effect to be negative (congestion): a large audience may spark spam or come with a loss of a community feeling through reduced chances of directly engaging with the creator. In such settings, an optimal monetization scheme permits users to pay extra to exclude others from consumption. 
While the direct profit network effects do not by themselves trigger novel payment schemes, they are a crucial ingredient to understand this market, where they reflect additional business opportunities that only emerge for popular creators with sufficiently large audiences.
All ``Top Creators 2023" listed by \cite{forbes2023} make a significant fraction of their income through merchandise, advertisement deals or other partnerships that only arise through their fame. This feature of our model explains why, for instance, a blogger or podcaster may provide content for free to attract a large audience of which some then buy their book or their designed clothes.

While our suggested indirect implementations are tailored to and inspired by the creator economy, our general model fits a plethora of other settings, in which a seller offers a club good with network effects. For example, consider an organizer of a farmers' market offering licenses to operate a market stall. A farmer's valuation for a license may exhibit non-monotone network effects. First, value network effects might be positive because there must be some variety to attract shoppers. However, these value network effects quickly become negative because a larger number of competitors reduces each individual farmer's profits.
The organization of the market has a fixed cost, and, while the marginal cost of adding another stall may be negligible, it can be incorporated as a profit network effect, which can also reflect any other cost or benefit of having a larger farmers' market.


\textbf{Literature:} \citet{cornelli1996} considers our baseline model without any network effects. Due to this relation, our indirect implementation through voluntary payment mechanisms also extends the payment scheme proposed in her paper. She considers a monopolistic mechanism designer who can produce a good at a fixed cost and zero marginal cost. By rewriting the profit maximization problem, we essentially model the network effects as (possibly negative) marginal costs. However, in contrast to \citet{schmitz1997}, who extends the cost function of \citet{cornelli1996} to agent specific but constant costs, the ``costs" in our setting depend on the number of consuming buyers. The network effect on the buyer's value engenders in a type-dependent cost-benefit analysis. The similarities with these papers also connect our paper to the literature on crowdfunding \citep{belleflamme2015,strausz2017,ellman2019,deb2019} and serial cost sharing \citep{moulin1992,moulin1994}. To the best of our knowledge, this literature has not studied how to sell network goods, yet.

By modelling digital content as a club good \citep{buchanan1965}, our work also relates to the mechanism-design literature on excludable public goods \citep{deb1999,hellwig2003,hellwig2005,norman2004,hellwig2007,bierbrauer2011}, where the goal is efficient provision rather than profit maximization. \cite{birulin2006} considers public goods with congestion, but, in contrast to us, models the congestion as a capacity constraint rather than incorporating it directly in the agents' payoff function. 

In industrial organization, it is standard  \citep[since][]{rohlfs1974} to model network goods by making agents' consumption value dependent on the number of other consumers. Mostly, this literature considers positive network effects, where the typical examples are fax machines or telephones. \citet{imas2024} provide evidence that consumers' valuations for the consumption of a good can also be increasing in others' unmet desires. That is, all else equal the willingness-to-pay increases when other consumers are excluded from the market.

Mechanism design with allocation externalities was also studied by \cite{jehiel1996}, but they are concerned with the externality on agents who did not acquire the good rather than joint consumption.  \cite{akbarpour2024} study vaccine allocation, where the externality also effects people that do not consume the good. 
The externality of the good in \cite{csorba2008,kang2020,ostrizek2023,pai2022} depends on the total production of the good in the economy. 
The consumption value in \cite{segal1999} depends on other agents' trades. While the externality in our model only depends on the number of consumers, the externality in \cite{dworczak2024} depends on the composition of the consumer set.


\section{Model} \label{sec:model}

\textbf{Players and outcomes:} A monopolistic seller (mechanism designer) can produce a non-rival and excludable good at cost $c$.
She faces $N$ buyers $i \in \mathcal N$, and she designs the mechanism that determines the outcome. Formally, an outcome $o=(q_{i},  m_{i})_{i\in \mathcal N}$ specifies for each buyer $i$ whether he can consume the good, $q_{i} \in \{0,1\}$, and his payment $m_{i} \in \mathbb R$.

\textbf{Profit:}
For a given outcome, we call the subset of buyers that consume the good the consumer set, $J = \{i: q_i=1\} \subseteq \mathcal N$. 
The seller maximizes her expected profit, and she only incurs cost $c$ when the good is provided, i.e., when there is at least one consumer.
For a given outcome, the seller's profit is
\eq{ \label{eq:designer-payoff}
 \sum_{i=1}^{N}  m_i  + \varphi (k) - c \mathbbm{1}_{k>0},
}
where function $\varphi: \{0,1,\dots,N \} \to \mathbb R$ maps a number of consumers $k=|J|$ into direct profit network effects that can be positive or negative depending on the application.

\textbf{Types:} Buyer $i$ privately learns his value type $\theta_{i}$, an iid draw from a commonly known distribution with cdf $F$, continuous and positive density $f$, and  support $\Theta :=[0,\overline \theta]$. 
A type profile is  denoted by $\boldsymbol{\theta} := (\theta_{i})_{i\in \mathcal N} \in  \Theta^N$, and we sometimes use the notation $\boldsymbol{\theta}=(\theta_{i},\boldsymbol{\theta}_{-i})$. The joint distribution of $\boldsymbol{\theta}$ is  given by $G(\boldsymbol{\theta}) = \Pi_{j \in \mathcal N}  F (\theta_j)$, and we define $G_{-i}(\boldsymbol \theta_{-i})$ analogously.

\textbf{Valuations:} Given a consumer set of size $|J|=k$, buyer $i$'s valuation for the good is $v(\theta_i,k)$, i.e., the utility of buyer $i$ is given by
\ea{
q_i v(\theta_i,k) -  m_i.
}
We allow for positive and negative value network effects, i.e., $v$ can be increasing, decreasing or non-monotone in $k$, but we assume these network effects go in the same direction for all types, i.e., $\sign(v(\theta_i,k)-v(\theta_i,k')) = \sign(v(\theta'_i,k)-v(\theta'_i,k'))$ for all $k',k, \theta_i,\theta'_i$.
Moreover, $v$ increasing in $\theta_i$ with a strictly positive derivative with respect to the type, $v_1(\theta_i,k):=\frac{\partial v}{\partial \theta_i}(\theta_i,k)>0$ for all $\theta_i$ and $k$, and we impose the following single-crossing condition
\eq{ \label{eq:sc-cond}
\sign \{ v_1(\theta_i,k) - v_1(\theta_i,k') \} =
\sign \{ v(\theta_i,k) - v(\theta_i,k') \} \quad \mbox{ for all } \theta_i, k\neq k',
}
which implies that the effect of a change in the number of consumers on the marginal valuation in terms of types goes in the same direction as the value network effect. That is, for positive network effects, larger types benefit more when the consumer set expands; for negative network effects, larger types lose more when the consumer set expands. 
Additionally, we assume
\ean{ 
&\sign  \left\{  \frac{\partial v_1}{\partial \theta_i}(\theta_i,k') -\frac{\partial v_1}{\partial \theta_i}(\theta_i,k) \right\} =\sign \{     v(\theta_i,k) - v(\theta_i,k') \}  \quad \mbox{ for all } \theta_i, k \neq k', \notag \\
&\mbox{ and } \quad \frac{1-F(x)}{f (x)} \leq \frac{1-F(y)}{f (y)} \quad \mbox{ for all } x>y, \label{eq:sc-cond2}
}
which, as we show in Lemma \ref{lem:sc-cond} in the appendix, imply that single-crossing also holds for virtual values. The virtual value of type $\theta_i$ in a consumer set of size $k$ is given by 
\eq{ \label{eq:virt-value}
\psi(\theta_{i},k) = v(\theta_{i},k) - \frac{1-F(\theta_i)}{f (\theta_i )} v_1(\theta_i,k),
}
where the latter part reflects the information rents needed to incentivize truthful type revelation. We assume it is strictly increasing in $\theta_i$, and, in line with \citet{myerson1981}, we call such environments regular. Because $v_1(\theta_i,k)\geq 0$ for all $k$ and $\theta_i$, the monotone hazard rate condition of \eqref{eq:sc-cond2} is sufficient for regularity when $v$ is concave in the type.

\textbf{Game:} The seller sets up an arbitrary (finite) game in which each buyer selects an action (plan) $\alpha_{i} \in \mathcal A_{i}$ with a strategy $\sigma_i: \Theta \to \mathcal A_i$. Let $\mathcal A := (\mathcal A_{i})_{i\in \mathcal N}$.
A deterministic outcome function $g: \mathcal A  \to \mathcal O$ maps an action profile into an outcome $o \in \mathcal O$, where $g_i(\boldsymbol{\alpha})$ is the final allocation decision and payment of buyer $i$ given all players' actions $\boldsymbol{\alpha}$. Moreover, each buyer must receive at least his outside option, which we normalize to a payoff of zero. By the revelation principle,\footnote{The classical revelation principle may not hold when restricting attention to deterministic mechanisms. However, this is not an issue in our setting with ex-post constraints. See \citet{jarman2017}. Moreover, we show that our regularity assumption on \eqref{eq:virt-value} implies that the optimal mechanism is indeed deterministic so that our restriction is without loss in our setting.} we can restrict attention to incentive-compatible direct revelation mechanisms (DRM) in our quest to find the optimal allocation because they can replicate any $\langle \mathcal A,g,\sigma\rangle$. Because DRM are often considered impractical, we also propose indirect mechanisms that implement this allocation.

\textbf{DRM:} In a deterministic DRM, each buyer $i$ reports his type, and functions 
$\langle q,m \rangle =  (q_i, m_i)_{i\in \mathcal N }$
determine the outcome for each combination of types, $q_i: \Theta^N \to \{0,1\}$ and $m_i: \Theta^{N} \rightarrow \mathbb R$.
For any deterministic DRM, we can define the number of consumers (size of the consumer set) as $k(\boldsymbol \theta):=\sum_{ i\in \mathcal N} q_i(\boldsymbol{\theta})$.

Given the other buyers' reported types $\boldsymbol {\widehat \theta}_{-i}$, the payoff of a buyer of type $\theta_i$ who reported $\widehat \theta_i$ to the DRM is 
\ea{
 u_i ( \widehat \theta_i,\boldsymbol {\widehat \theta}_{-i} | \theta_i)
= q_i( \widehat \theta_i,\boldsymbol {\widehat \theta}_{-i})  v(\theta_i,k( \widehat \theta_i,\boldsymbol {\widehat \theta}_{-i})) - m_i ( \widehat \theta_i,\boldsymbol {\widehat \theta}_{-i}) .
}
Note that buyer $i$'s utility depends on the final allocation through the size of the consumer set, which depends on the other buyers' reported types, but not their true types. That is, we are still in a private-value setting. Moreover, for fixed  $\boldsymbol {\widehat \theta}_{-i}$ and a given function $q$, buyer $i$'s report fixes the allocation (and hence $k$) deterministically. We impose ex-post (dominant-strategy) incentive and participation constraints when maximizing expected profit can be expressed as
\begin{align} \label{eq:IC}
U_i (\theta_i,\boldsymbol {\widehat \theta}_{-i}) = u_i( \theta_i,\boldsymbol {\widehat \theta}_{-i} | \theta_i) &\geq u_i( \widehat \theta_i ,\boldsymbol {\widehat \theta}_{-i} | \theta_i)  \quad &\forall i, \theta_i,\widehat \theta_i, \boldsymbol {\widehat \theta}_{-i}
\tag{DS-IC}\\
u_i( \theta_i, \boldsymbol {\widehat \theta}_{-i} | \theta_i) &\geq 0 \quad &\forall i, \theta_i, \boldsymbol {\widehat \theta}_{-i}.
\tag{IR} \label{eq:IR}
\end{align}
In contrast to a rival-goods problem, 
we do not have a restriction $\sum_i q_i(\boldsymbol \theta)\leq 1$
because the good can be consumed by all buyers at the same time such that the only feasibility constraint of our deterministic mechanism is $q_i (\boldsymbol {\theta}) \in \{0,1\}$ for all $\boldsymbol {\theta}$ and $i$.

The necessary and sufficient conditions for implementability under \eqref{eq:IC} are easier to verify than the corresponding cyclical monotonicity \citep[see, e.g.,][]{borgers2015} under a weaker Bayesian incentive constraint. Therefore, our approach is to first solve the easier and more constrained problem, and then we verify that our solution remains optimal when the constraint is relaxed.




\section{Analysis} \label{sec:analysis}

\textbf{Road map:} 
We split our analysis into two parts: first we focus on the direct implementation in dominant strategies, and we then turn to how this solution can be implemented indirectly.
After some preliminary observations implementability, we approach our problem in the usual fashion by first considering a relaxed problem, i.e., we maximize profits without the incentive-driven monotonicity constraints.
Next, we verify that our solution indeed satisfies these constraints and, hence, it also solves our original (more constrained) problem. Moreover, we show that relaxing \eqref{eq:IC} to a Bayesian incentive constraint leads to an identical solution. Finally, we discuss indirect implementations of our optimal allocation to rationalize commonly seen elements in real-life monetization schemes in the creator economy.

\subsection{Direct implementation in dominant strategies} \label{sec:direct}

\textbf{Implementability:}
 The constraint \eqref{eq:IC} implies that $q_i$ must be weakly increasing in type $\theta_i$ for all $\boldsymbol {\theta}_{-i}$ and that higher types must get ``better" consumer sets, whereas the transfers are pinned down by the familiar integral form \eqref{eq:integral-m} below. Consequently, \eqref{eq:IC} implies (IR) if the lowest type gets at least utility zero for all $\boldsymbol {\theta}_{-i}$. The following results are helpful in rewriting our problem.

\begin{lem} \label{lem:implementability}
A DRM $\langle q,m \rangle$ is incentive compatible if and only if for every $i$ and every $\boldsymbol{\theta}_{-i}$, 

(i) 
there is a type
$\underline x (\boldsymbol{\theta}_{-i})$ such that for all
$\theta_i > \underline x (\boldsymbol{\theta}_{-i}) > \theta'_i$:
\eq{ \label{eq:cutoff}
q_i (\theta_i, \boldsymbol{\theta}_{-i}) = 1 \mbox{ and } q_i (\theta'_i, \boldsymbol{\theta}_{-i}) = 0  ;
}
(ii) 
for all $\theta_i > \theta''_i > \underline x (\boldsymbol{\theta}_{-i})$,
\eq{ \label{eq:q=1}
v(\theta_i,k(\theta_i,\boldsymbol{\theta}_{-i})) -v(\theta''_i,k(\theta_i,\boldsymbol{\theta}_{-i})) \geq v(\theta_i,k(\theta''_i,\boldsymbol{\theta}_{-i})) -v(\theta''_i,k(\theta''_i,\boldsymbol{\theta}_{-i}));
}
(iii) 
for all $\theta_i >  \underline x (\boldsymbol{\theta}_{-i}) > \theta'_i$:
\eq{ \label{eq:integral-m}
\begin{aligned}
& m_i(\theta_i',\boldsymbol{\theta}_{-i}) &= m(0,\boldsymbol{\theta}_{-i})& , \\
  &m_i(\theta_i,\boldsymbol{\theta}_{-i}) &= m(0,\boldsymbol{\theta}_{-i}) &+ v(\theta_i,k(\theta_i, \boldsymbol{\theta}_{-i})) - \int_{\underline x (\boldsymbol{\theta}_{-i})}^{\theta_i} v_1(t,k(t,\boldsymbol{\theta}_{-i}))d t
\end{aligned}
}
\end{lem}
The above lemma gives familiar necessary and sufficient conditions \eqref{eq:cutoff} and \eqref{eq:integral-m} for incentive compatibility. Because we assume that value network effects go into the same direction for all types, they agree on the preference order over consumer set sizes.
The next lemma tells us that \eqref{eq:q=1} and the single-crossing condition \eqref{eq:sc-cond}
together imply that higher types must get a weakly ``better" consumer set size. We have to order $k$ because we allow non-monotone value network effects.\footnote{For instance, the linear $v(\theta_i,k)=g(k)\theta_i$ with a non-monotone $g$ satisfies single crossing.}

\begin{lem} \label{lem:sc}
Order consumer set sizes $k_1,\dots,k_n$ such that $v(\theta_i,k_{j+1})  \geq v(\theta_i,k_{j})$ for all $\theta_i$.
    A DRM $\langle q,m \rangle$ satisfies \eqref{eq:q=1} given any $\boldsymbol{\theta}_{-i}$ if and only if  the interval $(\underline x(\boldsymbol{\theta}_{-i}),\overline \theta]$ is partitioned by cutoffs $x_1(\boldsymbol{\theta}_{-i}) \leq \dots \leq x_n(\boldsymbol{\theta}_{-i})$ such that $k(\theta_i,\boldsymbol{\theta}_{-i}) = k_j$ if $\theta_i \in (x_j(\boldsymbol{\theta}_{-i}),x_{j+1}(\boldsymbol{\theta}_{-i})]$ with $x_{n+1}(\boldsymbol{\theta}_{-i})=\overline \theta$.
\end{lem}

Exploiting the integral form \eqref{eq:integral-m}, we can rewrite the seller's objective as
\eq{ \label{eq:virt-value-profit}
\int_{\Theta^{N}} \left( \sum_{i=1}^{N} \psi (\theta_i,k(\boldsymbol \theta))  q_i(\boldsymbol \theta) + \varphi (k(\boldsymbol \theta)) - c  \mathbbm{1}_{k(\boldsymbol \theta)>0}  \right)  d G (\boldsymbol{\theta}). 
}

\subsubsection{The relaxed problem} \label{sec:relaxed}

Our first step towards the optimal allocation of the relaxed problem is a characterization of the type profiles under which the good can be provided profitably. The second step clarifies the structure of the consumer set in these cases for any number of consumers, while the third step establishes the optimal allocation for all type profiles.

\textbf{1. Good provision:} In the relaxed problem, we maximize 
\eq{ \label{eq:type-by-type}
\sum_{i=1}^{N} \psi (\theta_i,k(\boldsymbol \theta))  q_i(\boldsymbol \theta)  + \varphi (k(\boldsymbol \theta))   - c  \mathbbm{1}_{k(\boldsymbol \theta)>0} 
}
separately for all possible type profiles $\boldsymbol \theta$, and we disregard the constraints \eqref{eq:cutoff} and \eqref{eq:q=1} of Lemma \ref{lem:implementability}.
Because the problem is linear in $q_i$, it follows immediately that, in the relaxed problem, our restriction to deterministic mechanisms is without loss, $q_i(\boldsymbol \theta) \in \{0,1\}$ for all $\boldsymbol \theta$.

Given a fixed type profile, the seller prefers to sell to consumer set $J\neq \emptyset$ rather than not providing the good at all if the following inequality holds
\eq{ \label{eq:cost-covered}
\begin{aligned}
\Psi (\boldsymbol \theta|J) & \geq C (|J|)         \\
\mbox{ with } \Psi (\boldsymbol \theta|J) := \sum_{i\in J} \psi (\theta_i,|J|) \quad  & \mbox{ and } \quad C (k) := c - (\varphi(k)- \varphi(0)).
\end{aligned}
}
This inequality simply expresses that providing the good is only profitable if the revenue extractable from consumer set $J$ exceeds the network effect adjusted cost for a consumer set of size $|J|$. The extractable revenue $\Psi (\boldsymbol \theta|J)$ is the sum of virtual values of admitted types, and $C(k)$ adjusts the total costs for profit network effects. For instance, a negative cost adjustment reflects positive profit network effects.

\begin{lem} \label{lem:aggr-provision}
Consider the relaxed problem and take any $\boldsymbol\theta$.
In optimum, the good is provided if and only if a set $ J \subseteq \mathcal N $ exists such that
\eqref{eq:cost-covered} holds.
\end{lem}
The above result gives a necessary and sufficient condition for good provision, but it does not delineate to whom the good shall be provided. The proof follows straightforwardly from the seller's type-by-type profit function \eqref{eq:type-by-type}.

\textbf{2. Consumer set structure and simple exclusion policies:}  Given some type profile $\boldsymbol \theta$, the following inequality states when providing the good to buyers $J$ is preferred over consumer set $J'\neq \emptyset$, 
\eq{
\label{eq:J>J'}
\Psi (\boldsymbol \theta|J) - \Psi (\boldsymbol \theta|J')  \geq    \varphi(|J'|)- \varphi(|J|).
}
Without loss of generality, let us re-label buyers in order of their (virtual) types, $\theta_i\geq \theta_{i+1}$, where regularity implies $\psi (\theta_i,k)\geq \psi(\theta_{i+1},k)$ for all $k$.
Since $(\varphi(k)- \varphi(k))=0$ for all $k\in\mathbb N$, we can immediately infer from \eqref{eq:J>J'} that out of all possible consumer sets of the same size $k$, the seller prefers $J_k=\{ j: j\leq k\}$ the most.
That is, in the relaxed problem, every optimal allocation that accepts $k$ buyers must accept the $k$ buyers with the highest (virtual) types.
Hence, we can restrict attention to such consumer sets, and we only need to find the optimal number of consumers $k^*$ for each type  realization $\boldsymbol{\theta}$.
Let us call an allocation rule with such a structure a simple exclusion policy. That is, given the labeling convention, $q$ is a simple exclusion policy for some $k^*$ if and only if
\eq{ \label{eq:agg-structure}
q_i(\boldsymbol{\theta}) = 1 \iff i \leq k^*.
}

\begin{lem} \label{lem:aggr-structure}
Consider the relaxed problem and take any $\boldsymbol\theta$.
If the good is produced, the optimal allocation rule is a simple exclusion policy for some $k^*$.
\end{lem}

This result follows from the objective \eqref{eq:type-by-type} combined with the insights above. Having established the structure of the optimal consumer set, we can finalize the characterization of the optimal allocation by determining the optimal consumer set size $k^*$.

\textbf{3. Number of consumers:} Before we state the optimal simple exclusion policy, let us consider the effect on the seller's profit when $j$ buyers are added to a consumer set $J_k=\{ i: i\leq k\}$ that follows a simple exclusion policy for $k$. 
It is only profitable to add the next highest $j$ types if the additional extractable revenue from them exceeds the threshold
\eq{
\begin{aligned}
    \gamma (k,k+j ,\boldsymbol{\theta}_{\leq k}) :=  \overbrace{\varphi(k)- \varphi(k+j)}^{=C(k+j)-C(k)} + \sum_{i= 1}^{k} \big( \psi (\theta_i,k)-\psi (\theta_i,k+j) \big)
\end{aligned}
}
with $\boldsymbol{\theta}_{\leq k}=(\theta_1,\dots,\theta_k)$.
This threshold reflects two changes in the seller's profit that the revenue extractable from the $j$ additional buyers needs to compensate. First, the profit network effect changes by  $(\varphi(k)-\varphi(k+j))$. This first part can be used to represent the variable part of the total adjusted total cost, $C(k) = c + \gamma (0,k,\cdot)$ for all $k$.
Second, each buyer $i\leq k$ already admitted to the consumer set now garners value $v(\theta_i,k+j)\neq v(\theta_i,k)$. That is, due to the value network effects, the seller can extract either more or less value from the buyers $J_k$ already tentatively considered for her consumer set. While the sign of this effect only depends on whether we assume positive or negative network effects, its size also depends on $\boldsymbol{\theta}_{\leq k}$.

Given Lemmas \ref{lem:aggr-provision} and \ref{lem:aggr-structure}, our relaxed problem for a given $\boldsymbol{\theta}$ is reduced to
\eq{ \label{eq:reduced-max}
\begin{aligned}
\max_{k\in\mathcal{N}} \quad \sum_{i=1}^k \psi(\theta_i,k) + \varphi(k) \quad
\mbox{such that} \quad \sum_{i=1}^k \psi(\theta_i,k) \geq C(k).
\end{aligned}
}
Let $\overline k(\theta)$ be the solution to this problem if it exists, and let $\overline k(\theta)=0$ otherwise, i.e., if all $k\in \mathcal N$ violate the constraint. Next, 
we define
\eq{
\begin{aligned}
    &\overline q_i(\boldsymbol \theta) = \mathbbm{1}_{i \leq \overline k(\boldsymbol \theta)} \quad
\mbox{ and } \\
\label{eq:integral-m2}
    &\overline m_i(\boldsymbol \theta ) = v(\theta_i,\overline k(\boldsymbol \theta)) \: \overline q_i (\boldsymbol \theta ) - \int_0^{\theta_i} \overline q_i(x, \boldsymbol {\theta}_{-i}) v_1 (x, \overline k(x, \boldsymbol {\theta}_{-i})) d x.
\end{aligned}
}

\begin{lem} \label{lem:finder}
$\langle \overline q,\overline m \rangle$ as defined in \eqref{eq:integral-m2} is the solution to the relaxed problem.
\end{lem}

\textbf{Illustration of Lemmas \ref{lem:aggr-provision} to \ref{lem:finder}:} Figure \ref{fig:ex1} summarizes and illustrates our results in a two-buyer example, which we will pick up again in Figure \ref{fig:allocation}. Panel \ref{subfig:provision-lemma} depicts Lemma \ref{lem:aggr-provision}: only for $(\theta_1,\theta_2)$-combinations in the top-right area a consumer set $J$ whose extractable revenue covers the adjusted cost $C (|J|)$ exists, and for all other type combinations production is not profitable. Here, we also see an inefficiency due to asymmetric information because the first-best solution is to provide the good to both buyers in the region north-east of the dotted line.
In Lemma \ref{lem:aggr-structure}, depicted in Panel \ref{subfig:structure-lemma}, we show that we can restrict attention to simple exclusion policies: in the provision area above the $45^{\circ}$-line, the seller never wants to exclude buyer 2. Lemma \ref{lem:finder} solves the relaxed problem by identifying the revenue maximizing consumer set when costs are covered.  Consider point $x=(\theta_1^x,\theta_2^x)$ in Panel \ref{subfig:finder}, where $\psi(\theta_1^x,1)+\varphi(1) \geq \varphi(0).$ Hence, accepting buyer 1 is profitable, but additionally including buyer 2 is not, because $\psi(\theta_2^x,2)< \gamma(1,2,\theta_1^x)$, implying
$\psi(\theta_1^x,2)+\psi(\theta_2^x,2)+\varphi(2) < \psi(\theta_1^x,1)+ \varphi(1)$. Since the cost $C(1)$ is covered, we found our solution $\overline k (x)=1$. Buyer 1 has the same type in type profile $y$ 
in Panel \ref{subfig:finder}, but adding the second buyer is profitable,
$\psi(\theta_1^y,2)+\psi(\theta_2^y,2)+\varphi(2) > \psi(\theta_1^y,1)+ \varphi(1)$.

\begin{figure}[h!]
\begin{center}
\subfloat[Lemma \ref{lem:aggr-provision}]{ \label{subfig:provision-lemma}
\scalebox{0.65}{
\begin{tikzpicture}
\draw[->] (-0.3,0)  -- (5.5,0) node[below]{$\theta_1$};
\draw (5,0) node[below]{$1$};
\draw (0,5) node[left]{$1$};
\draw[->] (0,-0.3) -- (0,5.5) node[left] {$\theta_2$};
\draw[very thick] (0.7*5,0) node[below]{$\frac{7}{10}$} -- (0.7*5,33/80*5) -- (33/80*5,0.7*5) -- (0,0.7*5)  node[left]{$\frac{7}{10}$};
\draw  (3.2, 4.5) node{Good provided};
\draw  (3.2, 3.9) node{$\{1 \}, \{2 \}$ or $\{1,2 \}$};
\draw  (1.75, 1.5) node{Good not provided};
\draw  (1.8, 0.9) node{$\emptyset$};
\draw[dotted] (0.4*5,0) -- (0,0.4*5);
\draw (0,5) -- (5,5) -- (5,0);
\end{tikzpicture}
}} 
\subfloat[Lemma \ref{lem:aggr-structure}]{ \label{subfig:structure-lemma}
\scalebox{0.65}{
\begin{tikzpicture}
\draw[->] (-0.3,0)  -- (5.5,0) node[below]{$\theta_1$};
\draw (5,0) node[below]{$1$};
\draw (0,5) node[left]{$1$};
\draw[->] (0,-0.3) -- (0,5.5) node[left] {$\theta_2$};
\draw[very thick] (0.7*5,0) node[below]{$\frac{7}{10}$} -- (0.7*5,33/80*5) -- (33/80*5,0.7*5) -- (0,0.7*5)  node[left]{$\frac{7}{10}$};
\draw[dashed] (2.78125, 2.78125) -- (5,5);
\draw  (1.5, 4.3) node{$\{2\}$ or $\{1,2\}$};
\draw  (4.3, 3.2) node{$\{1,2\}$};
\draw  (4.3, 2.2) node{or};
\draw  (4.3, 1.2) node{$\{1\}$};
\draw  (1.5, 1.5) node{$\emptyset$};
\draw (0,5) -- (5,5) -- (5,0);
\end{tikzpicture}
}}
\subfloat[Lemma \ref{lem:finder}]{ \label{subfig:finder}
\scalebox{0.65}{
\begin{tikzpicture}
\draw[->] (-0.3,0)  -- (5.5,0) node[below]{$\theta_1$};
\draw[->] (0,-0.3) -- (0,5.5) node[left] {$\theta_2$};
\draw[very thick] (0.7*5,0) node[below]{$\frac{7}{10}$} -- (0.7*5,33/80*5) -- (5,27/80*5);
\draw[very thick] (0,0.7*5)  node[left]{$\frac{7}{10}$} -- (33/80*5,0.7*5) -- (27/80*5,5);
\draw[very thick] (33/80*5,0.7*5) -- (0.7*5,33/80*5);
\draw  (4.25, 1) node{$\{1\}$};
\draw  (1, 4.25) node{$\{2\}$};
\draw  (4, 4) node{$\{1,2\}$};
\draw  (1.5, 1.5) node{$\emptyset$};
\draw (0,5) -- (5,5) -- (5,0);
 \draw (0.8*5,0.5) circle[radius=2pt];
  \fill (0.8*5,0.5)  circle[radius=2pt];
   \draw (0.8*5,3) circle[radius=2pt];
  \fill (0.8*5,3)  circle[radius=2pt];
  \draw[->] (0.5*5,1) node[left]{x}  -- (0.8*5-0.1,0.5);
  \draw[->] (0.5*5,2) node[left]{y}  -- (0.8*5-0.1,3-0.1);
\end{tikzpicture}
}}
 \caption{A depiction of how we solve the relaxed problem.} \label{fig:ex1}
\end{center}
\end{figure}
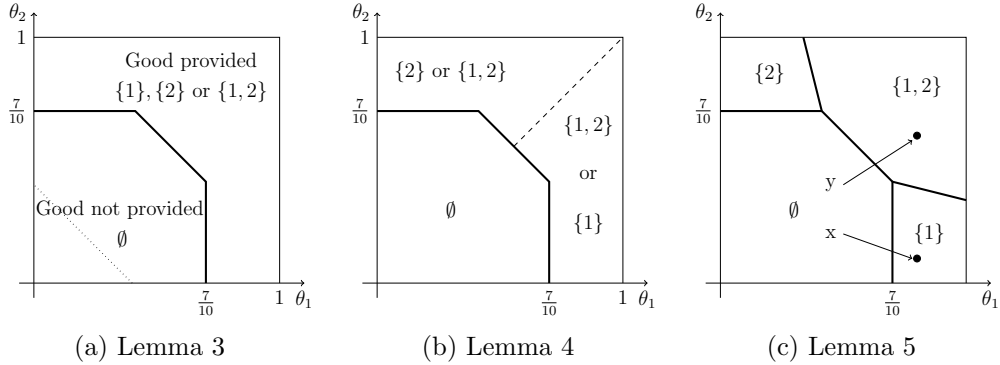

In the next section, we verify that our solution indeed satisfies the constraints that we ignored, and, as a result, the solutions to the relaxed problem and the constrained problem coincide.

\subsubsection{The full constrained problem}

Given our regularity assumption on virtual valuations \eqref{eq:virt-value} and our single-crossing assumption, the optimal allocation in the relaxed problem is indeed dominant-strategy incentive compatible and individually rational. Hence, it also solves the more constrained problem. First, $\overline q$ admits buyers to the consumer set in order of their virtual types, which under regularity coincides with the order of types. Therefore, any admitted buyer remains in the consumer set if his type is increased. Second, single crossing ensures that whenever two types obtain the good with different consumer set sizes, the larger type gets the consumer set size he prefers.

\begin{prop} \label{prop:optimal}
    In regular environments, $\langle \overline q,\overline m \rangle$ as defined in \eqref{eq:integral-m2} is the solution to the full constrained problem under \eqref{eq:IC}.
\end{prop}

\textbf{Discussion of regularity:} The optimality of $\overline q$ in the full constrained problem hinges on regularity in a fashion similar to the classical result by \citet{myerson1981}. With non-monotone virtual values, $\overline q$ would violate the monotonicity constraint. A solution to this problem would involve ironing virtual values separately for every consumer set size $k$ and then maximizing \eqref{eq:reduced-max} with the ironed virtual values. Ironing (or bunching) implies that the types in the ironed region get the same contract $(\widetilde q,\widetilde m)$. While this insight implies that the optimal allocation rule is stochastic in Myerson's model with a single private good, this is not necessarily true in our model. For instance, when all network effects are weakly positive, the seller always wants to add all types in the bunching region whenever she wants to add one of them. That is, restricting attention to deterministic allocation rules is still without loss here. However, this is clearly not true when negative network effects are allowed.\footnote{The most obvious example is the private-good setting, one of the extreme benchmark cases in Figure \ref{fig:benchmarks}.}

\textbf{Discussion of single crossing:} Conditions \eqref{eq:sc-cond} and \eqref{eq:sc-cond2} guarantee that a single-crossing condition also holds for virtual values, \eqref{eq:sc-virt} of Lemma \ref{lem:sc-cond} in the appendix. Similar to regularity, these conditions ensure that information rents are well-behaved as already suggested by \eqref{eq:sc-cond2} being a condition on the hazard rate and the second derivative of $v(\cdot,k)$ with respect to the type. If an increase in $\theta_i$ leads the seller to prefer consumer set size $k$ over $k'$, it must be that buyer $i$ weakly prefers it, too. That is, this change to consumer set size $k$ must not only be because it saves information rent, but also because it generates more value. If this was not the case, the mechanism would not be incentive compatible.
The examples we discuss below have uniformly distributed types (which implies a well-behaved hazard rate) and a value function $v$ that in linear in types , $v(\theta_i,k)=g(k)\theta_i$ (so that second derivatives are zero).

\begin{figure}[h!]
\begin{center}
\subfloat[No network effects.]{ \label{subfig:noNE}
\scalebox{0.80}{
\begin{tikzpicture}
\draw[->] (-0.5,0)  -- (6,0) node[below]{$\theta_1$};
\draw (5,0) node[below]{$1$};
\draw (0,5) node[left]{$1$};
\draw[->] (0,-0.5) -- (0,6) node[left] {$\theta_2$};
\draw[dotted] (0,0) -- (5,5);
\draw[dashed] (2.5,0) node[below]{$c=\frac{1}{2}$}  -- (0,2.5) node[left]{$c$};
\draw[dash dot] (1.25,5) -- (5,1.25);
\draw[very thick] (3/4*5,0) -- (3.75,2.5) -- (5,2.5) -- (3.75,2.5) -- (2.5,3.75)  -- (2.5,5) -- (2.5,3.75) -- (0,3.75) node[left]{$\frac{3}{4}$};
\draw  (4.25, 1.25) node{$\{1\}$};
\draw  (1.25, 4.25) node{$\{2\}$};
\draw  (4, 4) node{$\{1,2\}$};
\draw  (1.5, 1.5) node{$\emptyset$};
\draw (0,5) -- (5,5) -- (5,0);
\end{tikzpicture}
}} \quad
\subfloat[No value network effects, constant profit network effects, $\varphi^b(1)=\varphi^b(2)$.]{ \label{subfig:NEconst}
\scalebox{0.80}{
\begin{tikzpicture}
\draw[->] (-0.5,0)  -- (6,0) node[below]{$\theta_1$};
\draw[->] (0,-0.5) -- (0,6) node[left] {$\theta_2$};
\draw[dotted] (3/4*5,0) -- (3.75,2.5) -- (5,2.5) -- (3.75,2.5) -- (2.5,3.75)  -- (2.5,5) -- (2.5,3.75) -- (0,3.75);
\draw[very thick] (0.7*5,0)  -- (0.7*5,0.5*5) -- (5,0.5*5);
\draw[very thick] (0,0.7*5) node[left]{$\frac{7}{10}$} -- (0.5*5,0.7*5) -- (0.5*5,5) ;
\draw[very thick] (0.5*5,0.7*5) -- (0.7*5,0.5*5) ;
\draw  (4.25, 1) node{$\{1\}$};
\draw  (1, 4.25) node{$\{2\}$};
\draw  (4, 4) node{$\{1,2\}$};
\draw  (1.5, 1.5) node{$\emptyset$};
\draw (0,5) -- (5,5) -- (5,0);
\draw[dashed] (2.5-0.2*5,0) node[below]{$c-\varphi^b(2)$}-- (0,2.5-0.2*5); 
\end{tikzpicture}
}}

\subfloat[No value network effects, variable profit network effects, $\varphi^c(1)< \varphi^c(2)$.]{ \label{subfig:NEvar}
\scalebox{0.80}{
\begin{tikzpicture}
\draw[->] (-0.5,0)  -- (6,0) node[below]{$\theta_1$};
\draw[->] (0,-0.5) -- (0,6) node[left] {$\theta_2$};
\draw[dotted] (0.7*5,0)  -- (0.7*5,0.5*5) -- (5,0.5*5);
\draw[dotted] (0,0.7*5) node[left]{$\frac{7}{10}$} -- (0.5*5,0.7*5) -- (0.5*5,5) ;
\draw[dotted] (0.5*5,0.7*5) -- (0.7*5,0.5*5) ;
\draw[very thick] (0.7*5,0) -- (0.7*5,0.45*5) -- (5,0.45*5);
\draw[very thick] (0,0.7*5) node[left]{$\frac{7}{10}$} -- (0.45 *5,0.7*5) -- (0.45*5,5);
\draw[very thick] (0.45*5,0.7*5) -- (0.7*5,0.45*5);
\draw  (4.25, 1) node{$\{1\}$};
\draw  (1, 4.25) node{$\{2\}$};
\draw  (4, 4) node{$\{1,2\}$};
\draw  (1.5, 1.5) node{$\emptyset$};
\draw (0,5) -- (5,5) -- (5,0);
\draw[dashed] (2.5-0.2*5,0) node[below]{$c-\varphi^c(2)$}-- (0,2.5-0.2*5); 
\end{tikzpicture}
}} \quad
\subfloat[Positive value and profit network effects.]{ \label{subfig:NE}
\scalebox{0.80}{
\begin{tikzpicture}
\draw[->] (-0.5,0)  -- (6,0) node[below]{$\theta_1$};
\draw[->] (0,-0.5) -- (0,6) node[left] {$\theta_2$};
\draw[dotted] (0.7*5,0) -- (0.7*5,0.45*5) -- (5,0.45*5);
\draw[dotted] (0,0.7*5) node[left]{$\frac{7}{10}$} -- (0.45 *5,0.7*5) -- (0.45*5,5);
\draw[dotted] (0.45*5,0.7*5) -- (0.7*5,0.45*5);
\draw[very thick] (0.7*5,0) -- (0.7*5,33/80*5) -- (5,27/80*5);
\draw[very thick] (0,0.7*5) -- (33/80*5,0.7*5) -- (27/80*5,5);
\draw[very thick] (33/80*5,0.7*5) -- (0.7*5,33/80*5);
\draw  (4.25, 1) node{$\{1\}$};
\draw  (1, 4.25) node{$\{2\}$};
\draw  (4, 4) node{$\{1,2\}$};
\draw  (1.5, 1.5) node{$\emptyset$};
\draw (0,5) -- (5,5) -- (5,0);
\draw[dashed] (2.5-0.2*5,0) node[below]{$c-\varphi^b(2)$}-- (0,2.5-0.2*5); 
\end{tikzpicture}
}}
 \caption{Optimal allocations when adding components of our model step by step. The lines are straight due to the assumed linearity of $\psi$ (linear $v$ and uniformly distributed types).} \label{fig:allocation}
\end{center}
\end{figure}
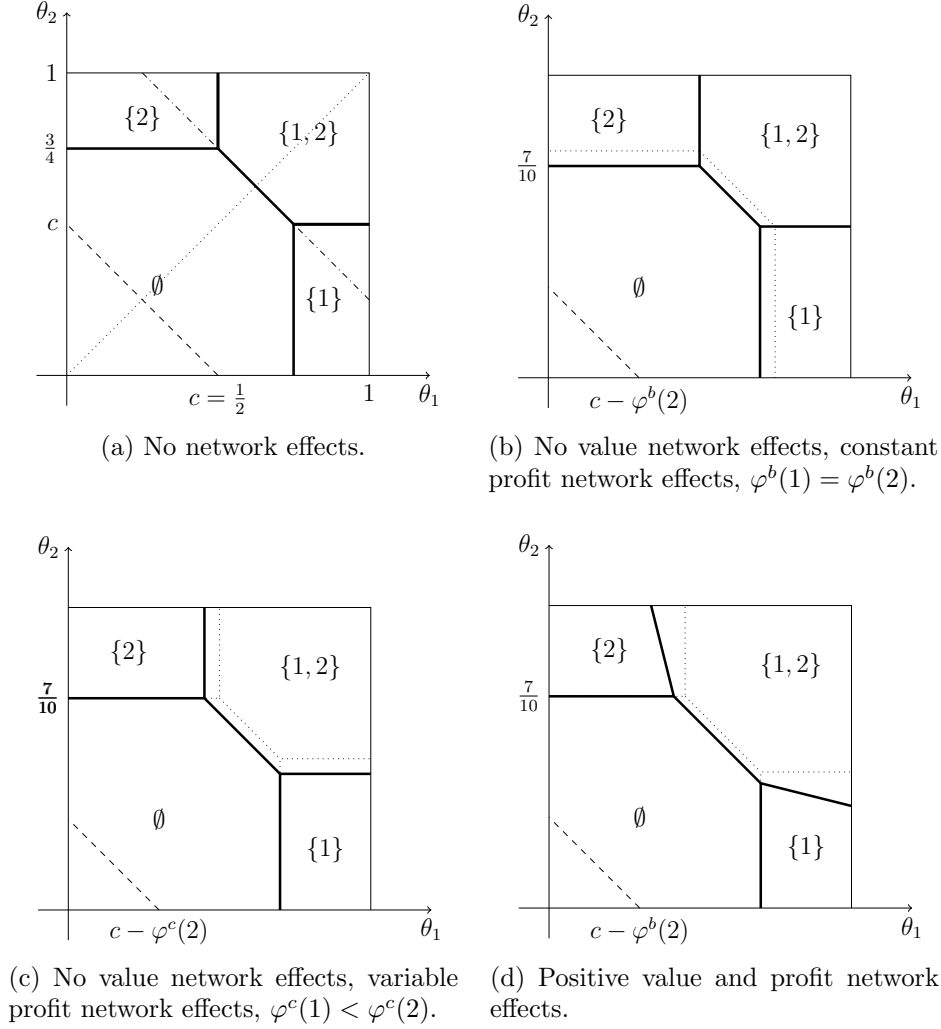

\textbf{Discussion of model components:}
Figure \ref{fig:allocation} juxtaposes the optimal allocations in four exemplary settings with $N=2$ buyers, uniformly distributed types on $[0,1]$, and cost $c=\nicefrac{1}{2}$. Step by step and starting from a setting without any network effects, we add a component of our model in each panel to discuss comparative statics. In Figure \ref{fig:benchmarks}, we focus on value network effects, and we discuss the intermediate cases between the benchmarks of private goods (strong negative value network effects) and public goods (strong positive value network effects).

\textbf{Panel \ref{subfig:noNE}} is copied from \citet[Figure 1]{cornelli1996}, who essentially solves our model without network effects, i.e., with $v^a(\theta_i,k)=\theta_i$ and $\varphi(k)=0$ for all $\theta_i$ and $k$. Here, the seller's first-best solution is to provide the good to both buyers for all $\boldsymbol{\theta}$ north-east of the dashed line defined by cost $c$ and not to provide it otherwise. 
Because incentive compatibility prevents the seller from extracting full surplus, the first-best allocation does not maximize profits when information rents are accounted for. For type combinations north-east of the dash-dotted line in Panel \ref{subfig:noNE}, the sum of virtual values covers the cost. Excluding types $\theta_i<\nicefrac{1}{2}$ with negative virtual values increases revenue, $\gamma^a(k-1,k,\boldsymbol{\theta}_{\leq k})=0$ for all $k, \boldsymbol{\theta}$. Hence, the good is provided if and only if the sum of non-negative virtual values exceeds the cost, $C^a(k)=c$ for all $k$. Contrary to the first-best allocation, sometimes only a single consumer may access the good. 
In contrast to the public-good case, the possibility to provide the good while excluding low types allows to maintain higher prices. In contrast to the private-good case, the non-rivalry allows the seller to accept all buyers that increase revenue. Thus, the non-rivalry combined with the fixed cost $c$ creates a positive externality among buyers even without network effects. All types $\theta_i>\nicefrac{1}{2}$ have a positive externality because they help to cover the seller's cost. This ``cost externality" is discussed later this section.

\textbf{Panel \ref{subfig:NEconst} and Panel \ref{subfig:NEvar}} incorporate profit network effects. In Panel \ref{subfig:NEconst}, they are constant for any provision, $\varphi^b(2)=\varphi^b(1)=\nicefrac{1}{10}> \varphi^b(0)=0$. The new optimal allocation is visualized by the thick black lines, whereas the dotted lines represent the optimal allocation from Panel \ref{subfig:noNE}. Incorporating these constant profit profit network effects, we see a uniform shift of the adjusted cost $C^b(k)=C^a(k)-\nicefrac{1}{10}$ and a shift in the thresholds for good provision to a single buyer and to two buyers. Panel \ref{subfig:NEvar} includes variable profit network effects, specifically $\varphi^c(2)=\nicefrac{1}{5} > \varphi^c(1)=\nicefrac{1}{10}> \varphi^c(0)=0$. There is an additional shift just for the provision for two buyers compared to Panel \ref{subfig:NEconst}. There is, however, no effect on the threshold of providing to one buyer only as the extractable valuation while providing to one buyer does not change. Only when provided to both buyers does the network effect allow for lower types to be included in the allocation.
To sum up, the profit network effects essentially lead to the model of \citet{cornelli1996} with a different cost function.

Straightforward graphical intuitions about the comparative statics in $c$ and $\varphi$ generalize. In any panel of Figure \ref{fig:allocation}, a change in the cost $c$ moves each border (and only those) around the area of types for which the good is not provided. Formally, we see that only the constraint in \eqref{eq:reduced-max} is affected.
An increase in $\varphi(k)$ both weakens the constraint and boosts the revenue for a specific $k$. Both changes considered in the result below only shift profits without an impact on buyers' incentives. Therefore, the proofs follow from the arguments above and are omitted.

\begin{lem}
    a) A decrease in cost $c$ weakly expands the set of type profiles $\boldsymbol{\theta}$ for which $\overline k(\boldsymbol{\theta})>0$, i.e., for which the good is provided.\\
    b) An increase in $\varphi(k)$ weakly expands the set of type profiles $\boldsymbol{\theta}$ for which $\overline k(\boldsymbol{\theta})=k$.
\end{lem}

In \textbf{Panel \ref{subfig:NE}}, we add positive value network effects to the setting of Panel \ref{subfig:NEvar}, setting $v^d(\theta_i,k)=\nicefrac{(2+k)\theta_i}{3}$. That is, if only a single buyer consumes, his value is as in the previous panels, but if both buyers consume jointly, their values increase. The dashed line characterizes the optimal allocation from Panel \ref{subfig:NEvar}. Because the value network effects are positive, this line is shifted to the south-west. Additionally, the lines separating the allocation of providing to one instead of two consumers are tilted, leading to an expansion of the set of type profiles with joint consumption. In the other panels, a consumer type left or below of this line has an insufficient virtual value and is excluded in \citet{cornelli1996} purely because of incentives, i.e., to maintain lower information rents for higher types. In Panel \ref{subfig:NE}, however, the seller wants to include some of these types to increase the valuation of the other consumer.

\begin{figure}[h!]
\begin{center}
\subfloat[Strong negative value network effects essentially make the good rival, \\
$1\geq \pi > \frac{2}{3}$.]{ \label{subfig:myerson}
\scalebox{0.85}{
\begin{tikzpicture}
\draw[->] (-0.5,0)  -- (6,0) node[below]{$\theta_1$};
\draw (5,0) node[below]{$1$};
\draw (0,5) node[left]{$1$};
\draw[->] (0,-0.5) -- (0,6) node[left] {$\theta_2$};
\draw[very thick] (5/8*5,0) -- (5/8*5,5/8*5) -- (5,5) -- (5/8*5,5/8*5) -- (0,5/8*5);
\draw[dashed] (11/16*5,0) -- (11/16*5,11/16*5) -- (5,5) -- (11/16*5,11/16*5) -- (0,11/16*5);
\draw  (4.5, 1) node{$\{1\}$};
\draw  (1, 4.5) node{$\{2\}$};
\draw  (1.5, 1.5) node{$\emptyset$};
\draw[dotted] (3/4*5,0) node[below]{$\frac{3}{4}$} -- (3.75,2.5) -- (5,2.5) -- (3.75,2.5) -- (2.5,3.75)  -- (2.5,5) -- (2.5,3.75) -- (0,3.75) node[left]{$\frac{3}{4}$};
\draw (0,5) -- (5,5) -- (5,0);
\end{tikzpicture}
}} \quad \quad
\subfloat[Small negative value network effects, \\$\pi=\frac{5}{8}$.]{ \label{subfig:smallNeg} 
\scalebox{0.85}{
\begin{tikzpicture}
\draw[->] (-0.5,0)  -- (6,0) node[below]{$\theta_1$};
\draw (5,0) node[below]{$1$};
\draw (0,5) node[left]{$1$};
\draw[->] (0,-0.5) -- (0,6) node[left] {$\theta_2$};
\draw[very thick] (7/10*5,0) -- (7/10*5,19/30*5) -- (5,5*5/6);
\draw[very thick] (0,7/10*5) -- (19/30*5,7/10*5) -- (5*5/6,5);
\draw[very thick] (7/10*5,19/30*5) -- (19/30*5,7/10*5);
\draw[dashed] (137/199*5,0) -- (137/199*5,2745/4020*5) -- (5,199/202*5);
\draw[dashed] (0,137/199*5) -- (2745/4020*5,137/199*5) -- (199/202*5,5);;
\draw[dashed] (137/199*5,2745/4020*5) -- (2745/4020*5,137/199*5);
\draw  (4.5, 1) node{$\{1\}$};
\draw  (1, 4.5) node{$\{2\}$};
\draw  (1.5, 1.5) node{$\emptyset$};
\draw  (4.2, 4.2) node{$\{ 1,2 \}$};
\draw[dotted] (3/4*5,0) node[below]{$\frac{3}{4}$} -- (3.75,2.5) -- (5,2.5) -- (3.75,2.5) -- (2.5,3.75)  -- (2.5,5) -- (2.5,3.75) -- (0,3.75) node[left]{$\frac{3}{4}$};
\draw (0,5) -- (5,5) -- (5,0);
\fill (0.6*5,3.75)  circle[radius=2pt];
\draw[dotted] (0.6*5,3.75) -- (0.6*5,0) node[below]{$y$};
\fill (0.8*5,3.75)  circle[radius=2pt];
\draw[dotted] (0.83*5,3.75) -- (0.83*5,0) node[below]{$x$};
\fill (0.95*5,3.75)  circle[radius=2pt];
\draw[dotted] (0.95*5,3.75) -- (0.95*5,0) node[below]{$z$};
\end{tikzpicture}
}}

\subfloat[Small positive value network effects, \\$\pi=\frac{3}{8}$.]{ \label{subfig:smallPos} 
\scalebox{0.85}{
\begin{tikzpicture}
\draw[->] (-0.5,0)  -- (6,0) node[below]{$\theta_1$};
\draw (5,0) node[below]{$1$};
\draw (0,5) node[left]{$1$};
\draw[->] (0,-0.5) -- (0,6) node[left] {$\theta_2$};
\draw[very thick] (5/6*5,0) -- (5/6*5,11/30*5) -- (5,3/10*5);
\draw[very thick] (0,5/6*5) -- (11/30*5,5/6*5) -- (3/10*5,5);
\draw[very thick] (5/6*5,11/30*5) -- (11/30*5,5/6*5);
\draw[dashed] (7/10*5,0) -- (7/10*5,19/30*5) -- (5,5*5/6);
\draw[dashed] (0,7/10*5) -- (19/30*5,7/10*5) -- (5*5/6,5);
\draw[dashed] (7/10*5,19/30*5) -- (19/30*5,7/10*5);
\draw  (4.5, 1) node{$\{1\}$};
\draw  (1, 4.5) node{$\{2\}$};
\draw  (1.5, 1.5) node{$\emptyset$};
\draw  (4.2, 4.2) node{$\{ 1,2 \}$};
\draw[dotted] (3/4*5,0) node[below]{$\frac{3}{4}$} -- (3.75,2.5) -- (5,2.5) -- (3.75,2.5) -- (2.5,3.75)  -- (2.5,5) -- (2.5,3.75) -- (0,3.75) node[left]{$\frac{3}{4}$};
\draw (0,5) -- (5,5) -- (5,0);
\end{tikzpicture}
}} \quad \quad
\subfloat[Strong positive value network effects  essentially make the good public, \\$\frac{1}{4} > \pi \geq 0$.]{ \label{subfig:public}
\scalebox{0.85}{
\begin{tikzpicture}
\draw[->] (-0.5,0)  -- (6,0) node[below]{$\theta_1$};
\draw (5,0) node[below]{$1$};
\draw (0,5) node[left]{$1$};
\draw[->] (0,-0.5) -- (0,6) node[left] {$\theta_2$};
\draw[very thick] (1/8*5,5) --  (5,1/8*5);
\draw[dashed] (1/6*5,5) --  (5,1/6*5);
\draw  (4.2, 4.2) node{$\{ 1,2 \}$};
\draw  (2, 2) node{$\emptyset$};
\draw[dotted] (3/4*5,0) node[below]{$\frac{3}{4}$} -- (3.75,2.5) -- (5,2.5) -- (3.75,2.5) -- (2.5,3.75)  -- (2.5,5) -- (2.5,3.75) -- (0,3.75) node[left]{$\frac{3}{4}$};
\draw (0,5) -- (5,5) -- (5,0);
\end{tikzpicture}
}}
 \caption{Optimal allocations compared to benchmark cases. If value network effects are sufficiently negative, the good is essentially private and provided to at most one buyer \citep{myerson1981}. If value network effects are sufficiently positive, the good is essentially public and provided to both or no buyers \citep{guth1986private}. The dotted lines represent the settings without network effects \citep{cornelli1996}.}\label{fig:benchmarks}
\end{center}
\end{figure}
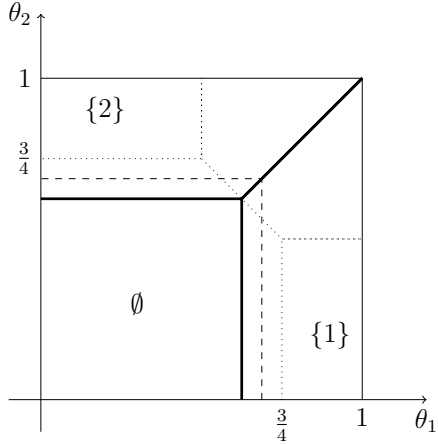
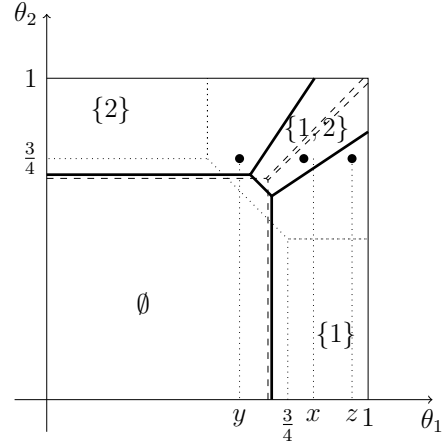
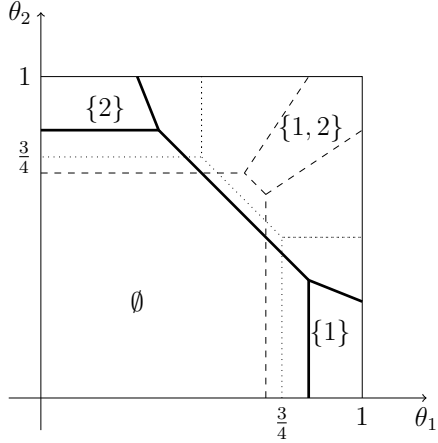
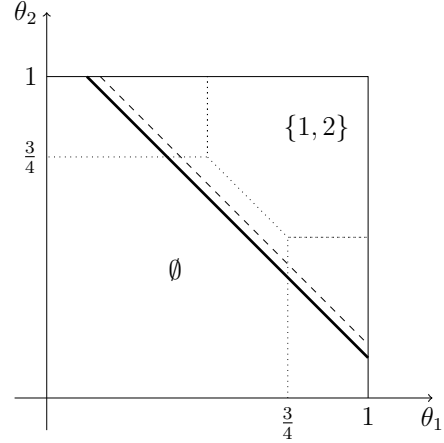

\textbf{Benchmark cases: Figure \ref{fig:benchmarks}} focuses on the value network effects (no profit network effects), and it shows how our model nests benchmarks from the literature. We assume uniformly distributed types on $[0,1]$, and cost $c=\nicefrac{1}{4}$ with the valuation function
\eq{ \label{eq:v-pi}
v(\theta_i,k)=\begin{cases}
			\pi \theta_i & \text{if }   k=1, \\
			(1-\pi)\theta_i & \text{if } k=2.
		 \end{cases}
}
That is, increasing $\pi$ makes consumption alone more valuable and consumption together less valuable.
The optimal allocation without network effects ($\pi=\nicefrac{1}{2}$, same as in Panel \ref{subfig:noNE}) is depicted by the dotted lines in each panel. As we elaborate below, the parameter $\pi$ can be seen as a degree of rivalry and excludability. While we use a two-buyer example for ease of exposition, there are several ways to extend valuation function \eqref{eq:v-pi} to cases with $n>2$.

\textbf{Panel \ref{subfig:myerson}} essentially represents the case of an indivisible, excludable, and rival good \citep{myerson1981}. The thick line in that panel represents the optimal allocation with $\pi=1$, i.e., a buyer only garners a payoff if he consumes the good alone and otherwise negative network effects destroy all value. In other words, for $\pi=1$, the good is fully rival.
Our optimal DRM of Proposition \ref{prop:optimal} collapses to a second-price auction with a reserve price.
A decrease in $\pi$ down to $\nicefrac{2}{3}$ (dashed line) does not qualitatively change the structure of the allocation. Only the reserve price changes as the good becomes less valuable when consumed alone while consumption together still destroys too much value that it is never optimal to have both buyers share the good. That is, although the good is not fully rival for such $\pi,$ it is still sufficiently rival that the fundamental structure of optimal selling mechanism is not affected.

\textbf{Panel \ref{subfig:smallNeg}} shows smaller negative network effects. For parameters $\pi$ slightly smaller than $\nicefrac{2}{3}$ (dashed line), it becomes optimal to have some types close to the 45-degree line share the good. However, if the types are too far apart, it is optimal to award the good only to the higher type. Further decreasing $\pi$ enlarges the area of type combinations for which the good is allocated to both buyers. The thick lines represent the optimal allocation with $\pi=\nicefrac{5}{8}$, and the dotted lines represent it for $\pi=\nicefrac{1}{2}$, the benchmark case of a non-rival and excludable good without network effects \citep{cornelli1996}. That is, in the interval $[\nicefrac{1}{2},1]$ the parameter $\pi$ can be seen purely as a degree of rivalry, with the extreme cases (fully rival and non-rival) at the endpoints, but the good is always fully excludable. 

\textbf{Panel \ref{subfig:smallPos}} shows the optimal allocation with small positive network effects with the thick line ($\pi=\nicefrac{3}{8}$), and it compares it to the allocations without network effects (dotted line) and negative network effects (dashed line).  An increase of $\pi$ above $\nicefrac{1}{2}$ does not affect the rivalry of the good, but it makes exclusion less favorable. In the extreme $\pi=0$, exclusion is not viable at all. That is, in the region $[0,\nicefrac{1}{2}]$ the parameter $\pi$ can be seen purely as a degree of excludability, with the extreme cases (fully excludable and non-excludable) at the endpoints, but the good is always fully non-rival.

In essence, \textbf{Panel \ref{subfig:public}} represents the case of a public good (non-excludable and non-rival). Here, a buyer only garners a payoff if no buyers are excluded and otherwise exclusion destroys all value,\footnote{To be precise, here it is not that the seller is unable to exclude buyers, but she does not want to because value to extract can only exits without exclusion. Our model would also allow to model exclusion costs directly through profit network effects $\varphi$.} i.e., $\pi=0$. Qualitatively, the optimal allocation looks the same for all $\pi$ smaller than $\nicefrac{1}{4}$ (dashed line), where the good is also only provided to both buyers or not at all. Our optimal DRM of Proposition \ref{prop:optimal} collapses to the (inefficient) private supply mechanism for public goods \citep{guth1986private}.

\textbf{Complements and substitutes:} A consequence of our dominant-strategy incentive constraint is that for each buyer $i$, the reported types of the other buyers determine a cutoff $z_i(\boldsymbol{\theta}_{-i})$, and $i$ gets to consume the good if and only if $\theta_i \geq z_i(\boldsymbol{\theta}_{-i})$. When this cutoff is weakly increasing everywhere, the allocation rule has substitutes \citep{milgrom2020clock,jarman2017ex}, and it has complements if the cutoff is weakly decreasing. That is, allocation rule $q$
\eq{
\begin{aligned}
    &\mbox{has substitutes: if } &\forall i,j, \theta_j>\theta'_j \quad &q_i(\theta'_j,\boldsymbol{\theta}_{-j})=0 &\Rightarrow \quad &  q_i(\theta_j,\boldsymbol{\theta}_{-j})=0 ,  \\
     & \mbox{and } &\forall i,j, \theta_j>\theta'_j \quad &q_i(\theta_j,\boldsymbol{\theta}_{-j})=1 &\Rightarrow \quad &  q_i(\theta'_j,\boldsymbol{\theta}_{-j})=1 ;  \\
    &\mbox{has complements: if } &\forall i,j, \theta_j>\theta'_j  \quad &q_i(\theta'_j,\boldsymbol{\theta}_{-j})=1 &\Rightarrow \quad &  q_i(\theta_j,\boldsymbol{\theta}_{-j})=1, \\
     & \mbox{and } &\forall i,j, \theta_j>\theta'_j \quad &q_i(\theta_j,\boldsymbol{\theta}_{-j})=0 &\Rightarrow \quad &  q'_i(\theta_j,\boldsymbol{\theta}_{-j})=0 ;  \\
\end{aligned}
}
In words, if an increase (decrease) in buyer $j$'s type can potentially kick another buyer $i$ out of the consumer set, buyers are substitutes (complements).

To illustrate the concept, consider the cutoff functions $z_2$ in Figure \ref{fig:benchmarks}. They are depicted by the curve such that all for all types $\theta_2$ above (below), buyer 2 is (not) in the consumer set. In the bottom two panels, where value network effects are positive, $z_2$ is weakly decreasing. That is, the allocation rule exhibits complements everywhere.
The virtual-value single-crossing condition ensures that this observation generalizes and leads to a monotone function $\overline k$, when profit network effects are weakly increasing.

\begin{lem} \label{lem:complements}
Suppose $\varphi(k)\geq \varphi(k')$ and $v(\theta_i,k)\geq v (\theta_i,k')$ for all $\theta_i$ and $k>k'$.
Then, the optimal allocation rule $\overline q$ has complements, and 
$\overline k(\theta_i,\boldsymbol{\theta}_{-i})$ is weakly increasing in each type $\theta_i$.
\end{lem}

In contrast, a corresponding statement with negative value network effects and substitutes is not true as illustrated by the non-monotonicity of $z_2$ in Panel \ref{subfig:smallNeg}. Special cases where the optimal allocation has substitutes are the private-good case (Panel \ref{subfig:myerson}) or the case with $c=0$ depicted in the right panel of Figure \ref{fig:c0}. The only mechanism that exhibits both complements and substitutes is the posted-price mechanism depicted in the left panel of Figure \ref{fig:c0}.

\textbf{Cost externalities:} More broadly, our settings can resemble that of a public good regardless of the direction of value network effects and despite the possibility of exclusion. The reason is that the fixed production cost $c$ creates a positive externality among the buyers, even when value network effects are slightly negative. To see this, consider the dotted lines in any panel of Figure \ref{fig:benchmarks}. They represent the setting without any network effects \citep{cornelli1996}. Here, buyer 1 may benefit from an increase in buyer 2's type when this increase pushes the extractable revenue above the cost. However, this is only relevant for types $\theta_1 \in [\nicefrac{1}{2},\nicefrac{3}{4}]$ because lower types never get the good, and higher types always get the good and are indifferent between consumption alone and together. With network effects, these boundaries vary with $\theta_2$. Even in the case of negative value network effects, the cost externality can create complements in the optimal allocation rule.
Holding a sufficiently large $\theta_2$ fixed in the allocation depicted in Panel \ref{subfig:smallNeg},
an increase in $\theta_1$ from $y$ to $x$ can increase $\overline k$ (it adds buyer 1), and a further increase to $z$ decreases $\overline k$ again (it kicks out buyer 2). Hence, the total externality is not monotone.

\textbf{Trivial economies:} In Lemma \ref{lem:always-never}, we determine conditions leading to a trivial outcome, i.e., never or always supplying the good. The first applies if the cost is sufficiently high, and the second applies if it is sufficiently low and there are sufficiently strong positive network effects. In the following, we call any economy satisfying the conditions of Lemma \ref{lem:always-never} ``trivial" and all other settings ``non-trivial." For part b) of the result below, we assume that $v(0,k)=0$ for all $k$.

\begin{lem} \label{lem:always-never}
The following settings constitute trivial economies.\\
a) The good is never provided, i.e., $(\overline q,\overline m)(\boldsymbol{\theta}) = (\mathbf{0},\mathbf{0})$ for all $\boldsymbol \theta$
if and only if 
\eq{ \label{eq:never-provide}
k \psi  (\overline \theta,k) < C (k) \quad \forall k.
}

b) The good is always provided for free to all, i.e., 
$(\overline q,\overline m)(\boldsymbol{\theta}) = (\mathbf{1},\mathbf{0})$ for all $\boldsymbol \theta$
if and only if 
\eq{ \label{eq:always-provide}
\begin{aligned}
&(i) &N \psi  (0,N) \geq C (N),  \quad \mbox{ and} \\
&(ii) &(N-k)\psi (0,N) \geq \gamma (k,N,\boldsymbol{0}_{\leq k})\quad \forall k.
\end{aligned}
} 
\end{lem}


\subsection{Indirect implementations in the creator economy} \label{sec:implementation}

Because DRM are rarely seen in practice, we propose simple indirect implementations of the optimal allocation. Specifically, we construct $\langle \mathcal A, g, \sigma \rangle$, a Bayesian game with action space $\mathcal A$, outcome function $g$, and a Bayesian Nash equilibrium $\sigma$ that implements the desired outcome for each type vector $\boldsymbol{\theta}.$ Because Bayesian implementation is less constrained, we first have to argue that $\langle \overline q, \overline m\rangle$ is still optimal.

In the following, we have an application to the creator economy in mind.
In light of the application, the terminology slightly changes: the buyers are called users, the seller is called creator, and the consumer set is called audience.
Our suggested implementations rationalize three commonly seen features of monetization schemes in the creator economy: donations (with an all-pay contribution mechanism), community subsidies (when value network effects are positive), and exclusivity bids (when value network effects are negative).

\subsubsection{Optimal interim allocation} \label{sec:interim}

Previously, we have dealt with the optimal allocation under the dominant-strategy incentive constraint \eqref{eq:IC}. In this subsection, we first establish that this allocation also solves our optimization problem under the weaker Bayesian incentive-compatibility constraint
\eq{
\label{eq:B-IC} \tag{B-IC}
\mathbb E_{\boldsymbol {\theta}_{-i}} \left[ u_i( \theta_i,\boldsymbol {\theta}_{-i} | \theta_i) \right] \geq 
\mathbb E_{\boldsymbol {\theta}_{-i}} [ u_i( \widehat \theta_i ,\boldsymbol {\theta}_{-i} | \theta_i) ] \quad \forall i, \theta_i,\widehat \theta_i,
}
where reporting the type truthfully does not have to be optimal given any of the other users' reports, but only has to maximize utility in expectation. Lemma \ref{lem:Bayes} below states that Proposition \ref{prop:optimal} extends to the optimal interim allocation.

Given a DRM $\langle q,m\rangle$, a user's type report amounts to selecting an interim allocation,
\eq{ \label{eq:expectations}
Q_i (\theta_i):= \mathbb{E}_{\boldsymbol{\theta_{-i}}} [q_i(\theta_i,\boldsymbol{\theta_{-i}})], 
\quad \mbox{ and } 
M_i (\theta_i):= \mathbb{E}_{\boldsymbol{\theta_{-i}}} [m_i(\theta_i,\boldsymbol{\theta_{-i}})] ,
}
so that $\langle \overline Q, \overline M\rangle$ corresponds to the expectations of $\langle \overline q, \overline m\rangle$. We can express the expected utility of a type-$\theta_i$ user as
\eq{ \label{eq:interim-utility}
U(\theta_i) = \mathbb E_{\boldsymbol {\theta}_{-i}} \left[ u_i( \theta_i,\boldsymbol {\theta}_{-i} | \theta_i) \right] 
= \sum_{k=1}^N Q_i^k(\theta_i) v(\theta_i,k) - M_i (\theta_i) ,
}
where $Q^k_i(\theta_i)$ is the probability that type $\theta_i$ consumes the good in a consumer set of size $k.$

\begin{lem} \label{lem:Bayes}
    In regular environments, $\langle \overline q,\overline m \rangle$ as defined in \eqref{eq:integral-m2} is the solution to the full constrained problem under the weaker constraint \eqref{eq:B-IC}.
\end{lem}

Our proof is a simple consequence of the fact that rewriting expected profit in terms of virtual values \eqref{eq:virt-value-profit} only requires the weaker \eqref{eq:B-IC}. We maximized this objective pointwise while ignoring the incentive conditions of Lemma \ref{lem:implementability}. However, these conditions are stricter than \eqref{eq:B-IC}. Hence, $\langle \overline q,\overline m \rangle$ is also optimal given the weaker Bayesian incentive constraint.

The following lemma shows that higher types pay more in expectation, and two types only have the same interim expected payment if they are indifferent over their respective contracts.
We will exploit the invertibility of the expected payment function to construct one of our indirect implementations.
Outside of the trivial cases, the good is only sometimes provided, and maybe only to some buyers. We define $\underline y$ as the lowest type to get the good for some type combination,
\eq{ \label{eq:lower-theta}
\underline y = \min_{\theta_i} \{\theta_i: \exists \boldsymbol{\theta_{-i}}:   \overline q_i (\theta_i,\boldsymbol{\theta_{-i}})=1 \} = \min_{\boldsymbol{\theta_{-i}}} \underline x (\boldsymbol{\theta_{-i}}).
}
All types $\theta_i < \underline y$ are always excluded. Incentive compatibility---i.e., Lemma \ref{lem:sc}---implies that type $\overline \theta$ must get the ``best contract" in the sense that he always consumes the good with the most preferred consumer set size among all types conditional on provision. However, there may be other types who get the same allocation for all $\boldsymbol{\theta}_{-i}$. Let the smallest of these types be
\eq{ \label{eq:upper-theta}
\overline y = \min_{\theta_i} \{\theta_i:    \overline q_i (\theta_i,\boldsymbol{\theta_{-i}})=\overline q_i (\overline \theta,\boldsymbol{\theta_{-i}})  \mbox{ and } \overline k (\theta_i,\boldsymbol{\theta_{-i}}) = \overline k (\overline \theta ,\boldsymbol{\theta_{-i}}) \forall \boldsymbol{\theta_{-i}} \},
}
which can be $\overline y=\overline \theta$, but this is not necessarily true. In optimum, all types $\theta_i \in [\overline y,\overline \theta]$ get the same contract.

\begin{lem} \label{lem:M-incr}
Consider any non-trivial economy.
The optimal interim expected transfer $\overline M_i$ is weakly increasing in $\theta_i$, and whenever $\overline M_i(x)= \overline M_i(y)$ for two types $x\neq y$, it must be that $\overline q_i(x,\boldsymbol{\theta}_{-i}) v(\theta_i,\overline k(x,\boldsymbol{\theta}_{-i})) =\overline q_i(y,\boldsymbol{\theta}_{-i}) v(\theta_i,\overline k(y,\boldsymbol{\theta}_{-i}))$ for all $\boldsymbol{\theta}_{-i}$ and $\theta_i\in [x,y]$.
\end{lem}

In an \textbf{all-pay contribution mechanism}, each user $i$ selects a price $p_i\in [0,\infty)$, and then an allocation decision is taken based on the price vector $\mathbf p = (p_i)_{i \in \mathcal N}$, while all users have to pay their selected price independent of the allocation decision. Formally, the action set of any all-pay contribution mechanism is a set of prices $\mathcal A_i=P_i=[0,\infty)$ so that $P =(P_i)_{i\in \mathcal N}$ , and the outcome is pinned down by an outcome function $g$ such that for any $\mathbf p$ and any user $i$, $g_i(\mathbf p) \in \{(0,p_i),(1,p_i)\}$. A (pure) strategy for a user $i$ in the all-pay contribution game induced by such a mechanism is a price-selection function $\rho_i: \Theta_i \to \mathcal A_i$ that maps a type into a price, and a strategy profile consists of all users' strategies ${\rho}=(\rho_i)_{i \in \mathcal N}.$ The following statement is about a specific all-pay contribution mechanism $g^*$ and a specific equilibrium $\rho^*$.

\begin{prop} \label{prop:allpay}
There is a  Bayesian Nash equilibrium in an all-pay contribution mechanism that implements the optimal allocation, $\langle  P, g^*,{\rho}^* \rangle$.
\end{prop}

The idea behind the construction of $g^*$ is simple. Essentially, it follows a reverse revelation principle. We know that the optimal DRM $\langle \overline q, \overline m\rangle$ is Bayesian incentive compatible and that $\overline M$ is invertible (it is either strictly increasing or, where it is constant, gives all types in this region the same contract). Hence, the game induced by outcome function 
\eq{
 g^*_i(\mathbf{p}) = ( q^*_i (\mathbf{p}),p_i) =
\begin{cases} 
(\overline q_i \big(\overline{\mathbf{M}}^{-1} (\mathbf{p})),p_i \big) \quad &\mbox{ if }
 p_i \in [\overline M_i(\underline y),\overline M_i(\overline y)], \\
(0,p_i) \quad &\mbox{ otherwise }
\end{cases}
}
has a Bayesian Nash equilibrium in which any type $\theta_i$ selects price $\rho^*_i(\theta_i) = \overline M_i(\theta_i)$ to get the following expected payoff
\ea{
u_i (\rho^*_i(\theta_i),\theta_i) &= \mathbb{E}_{\boldsymbol{\theta}_{-i}} \left[ q^*_i \big(\rho^*_{i}(\theta_i),{\rho^*}_{-i}(\boldsymbol{\theta}_{-i})\big) v \big(\theta_i,  k^*(\rho^*_{i}(\theta_i),\rho^*_{-i}(\boldsymbol{\theta}_{-i})\big) - \rho^*_i(\theta_i))
\right] \\
&= \mathbb{E}_{\boldsymbol{\theta}_{-i}} \left[ \overline q_i (\theta_{i},\boldsymbol{\theta}_{-i}) v (\theta_i, \overline k(\theta_i, \boldsymbol{\theta}_{-i})) - \overline m_i(\theta_{i},\boldsymbol{\theta}_{-i}))
\right],
}
where the incentive compatibility of $\langle \overline q,\overline m\rangle$ ensures that no type $\theta_i$ finds it profitable to deviate to selecting a price designated to another type. If user $i$ chooses a  price outside of the range of the interval $[\overline M_i(\underline y),\overline M_i(\overline y)],$ he never gets the good such that these deviations are also not profitable.
Importantly, as is customary in mechanism design, we allow the designer to select her preferred equilibrium $\rho^*$, and we do not worry about equilibrium multiplicity in general or possible bad equilibria more specifically. Settings with network externalities are prone to equilibrium multiplicity, see, e.g., \cite{halac2024}.

\textbf{Cost externality$\to$ voluntary payments:} Self-selected contributions are a common feature in monetization schemes in the creator economy, e.g., ``cheering" on Twitch or ``donations" on other platforms. Here, a high-type user volunteers to pay more than other users to jointly consume exactly the same content. Through the lens of our model, the rationale behind this behavior is  to increase the probability that the good is provided. That is, next to the transfer, also the threat of not producing is used to incentivize users.  The cost externality discussed in the previous section is the reason for a violation of a ``law of one price." In the benchmark of \cite{myerson1981}, there is no joint consumption and thus no cost externality such that the above implementation corresponds to an all-pay auction that awards the good only to the highest bidder. Under high rivalry and excludability, a user selects a higher price to increase his own allocation probability and to decrease the others'. In contrast, under low rivalry and excludability, a user selects a higher price to increase both his own and the others' allocation probability.

Often users only have to pay their contribution when they actually consume the good. In some settings, a similar construction is possible, in which the user only has to pay if he is a member of the audience, i.e., we reverse engineer the optimal allocation and payments through the equation $M_i(\theta_i)=Q_i(\theta_i) p_i$ rather than $M_i(\theta_i)= p_i$. For  instance, \cite{cornelli1996} suggests such a scheme or, as another example, the optimal allocation in \cite{myerson1981} can be implemented with a first-price auction rather than an all-pay auction (or the second-price auction essentially suggested by Proposition \ref{prop:optimal}). However, such an implementation is not always feasible as $\nicefrac{M_i}{Q_i}$ may not be invertible.

For the reason above, pure all-pay contribution mechanisms may appear unconventional. However, we can further tweak the payment rule to capture realistic features. Although the following two examples are not special cases of all-pay contribution mechanisms, the way how we decentralize the allocation mechanism is an offspring of the idea that self-selected contributions determine the allocation, where we just add more meaning to these contributions.

\begin{figure}[h!]
\captionsetup[subfloat]{farskip=30pt,captionskip=20pt}
\begin{center}
\subfloat{
\adjustbox{valign=b}{
\label{subfig:M_benchmark}
\scalebox{0.65}{
\begin{tikzpicture}[
declare function={
    M1(\x)= 25/128+\x^2/2;
    M0(\x)= -(1/128)+\x^2/2;
  }
  ]
\begin{axis}[%
            domain=0:1,
            yscale=1,
            axis lines=middle,
            xlabel={\large $\theta_i$},
            ylabel={\large $M_i$},
                       ymin=0,
            ymax=0.75,
            xmin=-0.1,
            xmax=1.1,
           samples=10,
             xlabel style={below right},
              xtick={0,1/8,5/8,1},
    xticklabels={0,$x_{\pi=0}$,$x_{\pi=1}$,1},
        ]
     \addplot [thick,domain=5/8:1] {M1(x)};
          \draw [dotted] (72.5,28) -- (72.5,400);
          \addplot [thick,domain=1/8:1] {M0(x)};
    \end{axis}
\end{tikzpicture}
}
}} 
\subfloat{
\adjustbox{valign=b}{
\label{subfig:Q_benchmark}
\scalebox{0.65}{
\begin{tikzpicture}[
declare function={
    Q1(\x)= \x;
     Q0(\x)= -(1/8)+\x;
  }
  ]
\begin{axis}[%
            domain=0:1,
            yscale=1,
            axis lines=middle,
            xlabel={\large $\theta_i$},
            ylabel={\large $Q_i$},
                       ymin=0,
            ymax=1,
            xmin=-0.1,
            xmax=1.1,
           samples=10,
             xlabel style={below right},
               xtick={0,1/8,5/8,1},
    xticklabels={0,$x_{\pi=0}$,$x_{\pi=1}$,1},
        ]
     \addplot [thick,domain=5/8:1] {Q1(x)};
     \addplot [thick,dashed,domain=1/8:1] {Q0(x)};
          \draw [dotted] (72.5,0) -- (72.5,62);
          
    \end{axis}
\end{tikzpicture}
}
}}

\subfloat{
\adjustbox{valign=b}{
\label{subfig:M_neg}
\scalebox{0.65}{
\begin{tikzpicture}[
declare function={
    Mc(\x)= ((900*\x*\x-361)/1920;
   Mc2(\x)= (700*\x*\x+269)/1920;
   Mc3(\x)= (1/480)*(40*(\x*\x)+161);
  }
  ]
\begin{axis}[%
            domain=0:1,
            yscale=1,
            axis lines=middle,
            xlabel={\large $\theta_i$},
            ylabel={\large $M_i$},
                       ymin=0,
            ymax=0.75,
            xmin=-0.1,
            xmax=1.1,
           samples=10,
             xlabel style={below right},
              xtick={0,19/30,7/10,5/6,1},
    xticklabels={0,$x$,$y$,$z$,1},
        ]
     \addplot [thick,domain=19/30:7/10-0.00001] {Mc(x)};
     \addplot [thick,domain=7/10:5/6-0.00001] {Mc2(x)};
     \addplot [thick,domain=5/6:1] {Mc3(x)};
          \draw [dotted] (80,50) -- (80,320);
    \end{axis}
\end{tikzpicture}
}
}}
\subfloat{
\adjustbox{valign=b}{
\label{subfig:Q_neg}
\scalebox{0.65}{
\begin{tikzpicture}[
declare function={
    Q1a(\x)=0*\x;
    Q1bc(\x)=1/6*(1+4*\x);
    Q2a(\x) =-(19/12)+(5*\x)/2;
    Q2b(\x) = 5/12*(-1+2*\x);
    Q2c(\x) = 1/6*(5-4*\x);
  }
  ]
\begin{axis}[%
            domain=0:1,
            yscale=1,
            axis lines=middle,
            xlabel={\large $\theta_i$},
            ylabel={\large $Q_i$},
                       ymin=0,
            ymax=1,
            xmin=-0.1,
            xmax=1.1,
           samples=10,
             xlabel style={below right},
              xtick={0,19/30,7/10,5/6,1},
    xticklabels={0,$x$,$y$,$z$,1},
        ]
     \addplot [thick,domain=19/30:7/10-0.00001] {Q1a(x)};
     \addplot [thick,domain=7/10:5/6-0.00001] {Q1bc(x)};
     \addplot [thick,domain=5/6:1] {Q1bc(x)};
     \addplot [thick,dashed,domain=19/30:7/10-0.00001] {Q2a(x)};
     \addplot [thick,dashed,domain=7/10:5/6-0.00001] {Q2b(x)};
     \addplot [thick,dashed,domain=5/6:1] {Q2c(x)};
     \addplot [thick,dash dot,domain=7/10:1] {1};
          \draw [dotted] (80,0) -- (80,63);
          \draw [dotted] (93,0) -- (93,28);
    \end{axis}
\end{tikzpicture}
}
}}

\subfloat{
\adjustbox{valign=b}{
\label{subfig:M_pos}
\scalebox{0.65}{
\begin{tikzpicture}[
declare function={
    Ma(\x)= ((1/128)*(100*\x*\x-9);
   Ma2(\x)= (5*\x*\x)/16-7/960;
   Ma3(\x)= (1/480)*(24*(\x*\x)+139);
  },
  ]
\begin{axis}[%
            domain=0:1,
            yscale=1,
            axis lines=middle,
            xlabel={\large $\theta_i$},
            ylabel={\large $M_i$},
            ymin=0,
            ymax=0.75,
            xmin=-0.1,
            xmax=1.1,
            samples=166,
            xlabel style={below right},
            xtick={0,3/10,11/30,5/6,1},
    xticklabels={0,$x$,$y$,$z$,1},
        ]
     \addplot [thick,domain=3/10:11/30-0.00001] {Ma(x)};
     \addplot [thick,domain=11/30:5/6-0.00001] {Ma2(x)};
     \addplot [thick,domain=5/6:1] {Ma3(x)};
     \draw [dotted] (93.5,220) -- (93.5,325);
    \end{axis}
\end{tikzpicture}
}
}}
 \subfloat{
\adjustbox{valign=b}{
\label{subfig:Q_pos}
\scalebox{0.65}{
\begin{tikzpicture}[
declare function={
    Q2a(\x)= -(3/4)+(5*\x)/2;
    Q2b(\x)= -(1/5)+\x;
    Q2c(\x)= 1+1/10*(-7+4*\x);
    Q1ab(\x)= 0*\x;
    Q1c(\x)= 1/10*(7 - 4*\x);
  },
  ]
\begin{axis}[%
            domain=0:1,
            yscale=1,
            axis lines=middle,
            xlabel={\large $\theta_i$},
            ylabel={\large $Q_i$},
            ymin=0,
            ymax=1,
            xmin=-0.1,
            xmax=1.1,
            samples=166,
            xlabel style={below right},
            xtick={0,3/10,11/30,5/6,1},
    xticklabels={0,$x$,$y$,$z$,1},
        ]
     \addplot [thick,domain=3/10:5/6] {Q1ab(x)};
     \addplot [thick,domain=5/6:1] {Q1c(x)};
     \addplot [thick,dashed,domain=3/10:11/30-0.00001] {Q2a(x)};
     \addplot [thick,dashed,domain=11/30:5/6-0.00001] {Q2b(x)};
     \addplot [thick,dashed,domain=5/6:1] {Q2c(x)};
     \addplot [thick,dash dot,domain=5/6:1] {1};
               \draw [dotted] (47,0) -- (47,16);
          \draw [dotted] (93,0) -- (93,63);
    \end{axis}
\end{tikzpicture}
}
}}
  \caption{Interim expected payments (left) and allocation probabilities (right, dashed: consumption alone, through: jointly, dash-dotted: total) for settings with varying $\pi$ (1st: $\pi=0$ and $\pi=1$, 2nd: $\pi=\nicefrac{5}{8}$, 3rd: $\pi=\nicefrac{3}{8}$). The corresponding optimal allocations are depicted in Figure \ref{fig:benchmarks}.}\label{fig:MQ}
\end{center}
\end{figure}
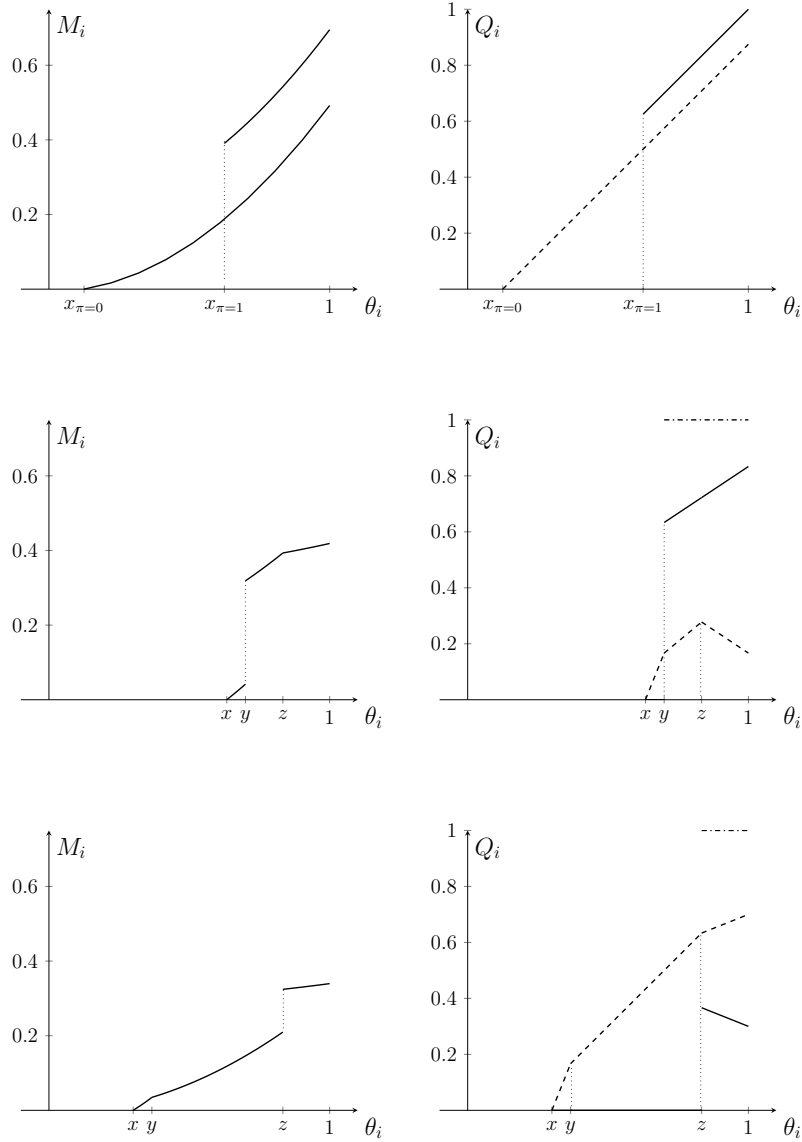

\begin{figure}[h!]

\captionsetup[subfloat]{farskip=30pt,captionskip=20pt}
\renewcommand{\arraystretch}{1.2}

\begin{center}

\subfloat[Implementation with voluntary payments and community gifts in a setting with positive network effects.]{ 
\adjustbox{valign=b}{
\label{subfig:implement_pos} 
\scalebox{0.85}{
\begin{tabular}{p{0.1\textwidth}p{0.15\textwidth}cp{0.2\textwidth}p{0.15\textwidth}}
type $\theta_i \in$ &action $(p_i,s_i)$ &&& outcome $(q_i,k,m_i)$\\
\hline
[0,x]  & (0,0)                                     &\tikzmark{posa1}&               &  \tikzmark{posa2} (0,$\cdot$,0)  \\
&&&&\\
(x,y]  & $(\underline p,s_i)$                                 & \tikzmark{posb1}&  \tikzmark{posb2} $s_j<s(s_i)$ \tikzmark{posb3} &  \tikzmark{posb4} $(0,\cdot,0)$  \\
          & $s_i<0$            &               & \tikzmark{posb5} $s_j\geq s(s_i)$ \tikzmark{posb6} & \tikzmark{posb7} ($1,2,\underline p+s_i$)  \\
&&&&\\
$(y,z]$  & $(p_i,0)$                                 &\tikzmark{posc1} &\tikzmark{posc2} $p_j < p(p_i)$ \tikzmark{posc3}  &  \tikzmark{posc4} $(0,\cdot,p_i$)  \\
          & $p_i \in (\underline{p},\overline{p})$ &                    &\tikzmark{posc5}  $p_j\geq p(p_i)$ \tikzmark{posc6} & \tikzmark{posc7} ($1,2,p_i$)  \\
&&&&\\
$(z,1]$  & $(\overline{p},s_i)$                         & \tikzmark{posd1} & \tikzmark{posd2} $s_i < s(s_j)$ \tikzmark{posd3} &  \tikzmark{posd4} $(1,1,\overline p+s_i$)  \\
          & $s_i >0$       &              &\tikzmark{posd5}  $s_i\geq s(s_j)$ \tikzmark{posd6} & \tikzmark{posd7} $(1,2,\overline p+s_i$)  \\
\end{tabular}
\begin{tikzpicture}[overlay, remember picture, shorten >=.5pt, shorten <=.5pt, transform canvas={yshift=.25\baselineskip}]
\draw [->] ([yshift=0pt]{pic cs:posa1}) -- ({pic cs:posa2});
\draw [-] ([yshift=-10pt]{pic cs:posb1}) -- ([yshift=-5pt]{pic cs:posb2}) ;
\draw [->] ([yshift=-2pt]{pic cs:posb3}) -- ([yshift=-2pt]{pic cs:posb4}) ;
\draw [-] ([yshift=-10pt]{pic cs:posb1}) -- ({pic cs:posb5}) ;
\draw [->] ([yshift=-2pt]{pic cs:posb6}) -- ([yshift=-2pt]{pic cs:posb7}) ;
\draw [-] ([yshift=-10pt]{pic cs:posc1}) -- ([yshift=-5pt]{pic cs:posc2}) ;
\draw [->] ([yshift=-2pt]{pic cs:posc3}) -- ([yshift=-2pt]{pic cs:posc4}) ;
\draw [-] ([yshift=-10pt]{pic cs:posc1}) -- ({pic cs:posc5}) ;
\draw [->] ([yshift=-2pt]{pic cs:posc6}) -- ([yshift=-2pt]{pic cs:posc7}) ;
\draw [-] ([yshift=-10pt]{pic cs:posd1}) -- ([yshift=-5pt]{pic cs:posd2}) ;
\draw [->] ([yshift=-2pt]{pic cs:posd3}) -- ([yshift=-2pt]{pic cs:posd4}) ;
\draw [-] ([yshift=-10pt]{pic cs:posd1}) -- ({pic cs:posd5}) ;
\draw [->] ([yshift=-2pt]{pic cs:posd6}) -- ([yshift=-2pt]{pic cs:posd7}) ;
\end{tikzpicture}
}
}}

\subfloat[Implementation with subscription fees and exclusivity bids in a setting with negative network effects.]{ 
\adjustbox{valign=b}{
\label{subfig:implement_neg} 
\scalebox{0.85}{
\begin{tabular}{p{0.1\textwidth}p{0.15\textwidth}c p{0.2\textwidth}p{0.15\textwidth}}
type $\theta_i \in$ &action $(p_i,s_i)$ &&& outcome $(q_i,k,m_i)$\\
\hline
[0,x]  & (0,0)                                     &\tikzmark{nega1}&               &  \tikzmark{nega2} (0,$\cdot$,0)  \\
&&&&\\
(x,y]  & $p_i \in (0, \overline{p}]$                                 & \tikzmark{negb1}&  \tikzmark{negb2} $b_j<\beta_1(p_i,p_j)$ \tikzmark{negb3} &  \tikzmark{negb4} $(1,2,p_i)$  \\
          & $b_i =0$            &               & \tikzmark{negb5} otherwise \tikzmark{negb6} & \tikzmark{negb7} ($0,\cdot,p_i$)  \\
&&&&\\
&                      &                    &\tikzmark{negc8}  $b_j\geq \overline \beta_2(b_i,p_j)$ \tikzmark{negc9} & \tikzmark{negc10} ($0,1,\overline p$)  \\
$(y,z]$  & $p_i= \overline{p}$                                 &\tikzmark{negc1} &\tikzmark{negc2} $b_i > \underline \beta_2(b_j,p_j)$ \tikzmark{negc3}  &  \tikzmark{negc4} $(1,1,\overline p + b_i$)  \\
          & $b_i  > 0$ &                    &\tikzmark{negc5}  otherwise \tikzmark{negc6} & \tikzmark{negc7} ($1,2, \overline p$)  \\
&&&&\\
$(z,1]$  & $p_i = \overline{\overline{p}}$      & \tikzmark{negd1} & \tikzmark{negd2} $b_i >\beta_3(b_j,p_j)$ \tikzmark{negd3} &  \tikzmark{negd4} $(1,1,\overline{\overline p}+b_i$)  \\
          & $b_i >0$       &              &\tikzmark{negd5}  otherwise \tikzmark{negd6} & \tikzmark{negd7} $(1,2,\overline  {\overline p}$)  \\
\end{tabular}
\begin{tikzpicture}[overlay, remember picture, shorten >=.5pt, shorten <=.5pt, transform canvas={yshift=.25\baselineskip}]
\draw [->] ([yshift=0pt]{pic cs:nega1}) -- ({pic cs:nega2});
\draw [-] ([yshift=-10pt]{pic cs:negb1}) -- ([yshift=-5pt]{pic cs:negb2}) ;
\draw [->] ([yshift=-2pt]{pic cs:negb3}) -- ([yshift=-2pt]{pic cs:negb4}) ;
\draw [-] ([yshift=-10pt]{pic cs:negb1}) -- ({pic cs:negb5}) ;
\draw [->] ([yshift=-2pt]{pic cs:negb6}) -- ([yshift=-2pt]{pic cs:negb7}) ;
\draw [-] ([yshift=0pt]{pic cs:negc1}) -- ([yshift=-0pt]{pic cs:negc2}) ;
\draw [->] ([yshift=-0pt]{pic cs:negc3}) -- ([yshift=-0pt]{pic cs:negc4}) ;
\draw [-] ([yshift=-0pt]{pic cs:negc1}) -- ({pic cs:negc5}) ;
\draw [->] ([yshift=-0pt]{pic cs:negc6}) -- ([yshift=-0pt]{pic cs:negc7}) ;
\draw [-] ([yshift=0pt]{pic cs:negc1}) -- ([yshift=-0pt]{pic cs:negc8}) ;
\draw [-] ([yshift=0pt]{pic cs:negc9}) -- ([yshift=-0pt]{pic cs:negc10}) ;
\draw [-] ([yshift=-10pt]{pic cs:negd1}) -- ([yshift=-5pt]{pic cs:negd2}) ;
\draw [->] ([yshift=-2pt]{pic cs:negd3}) -- ([yshift=-2pt]{pic cs:negd4}) ;
\draw [-] ([yshift=-10pt]{pic cs:negd1}) -- ({pic cs:negd5}) ;
\draw [->] ([yshift=-2pt]{pic cs:negd6}) -- ([yshift=-2pt]{pic cs:negd7}) ;
\end{tikzpicture}
}
}}

 \caption{Implementations for settings with positive (a) and negative (b) value network effects. 
 See the appendix for details on Example (a) and (b).}\label{fig:implementation}

\end{center}
\end{figure}

\textbf{Positive network effects$\to$ community gifts:} Twitch also allows gifted subscriptions (``community subs") by one user to others. In our model, users can have an intrinsic preference for joint consumption with a larger audience such that there is an implicit motive to subsidize other users. Figure \ref{subfig:implement_pos} illustrates how to employ subsidies (such as community gifts) in the 2-user environment with positive value network effects depicted in Panel \ref{subfig:smallPos}.
Here, a user's action is to choose $(p_i,s_i)$, a price $p_i$ and a subsidy $s_i$.
In the appendix, we jointly construct the outcome function and the equilibrium to implement the optimal allocation, where the types are partitioned into the regions $[0,x], (x,y], (y,z],$ and $(z,1],$ which are depicted in Panel \ref{subfig:M_pos}. 

Low-type users $\theta_i <x$ abstain by selecting $p_i=s_i=0$. Low intermediate types $\theta_i \in (x,y]$ select the minimum price $p_i=\underline p$ and request a price reduction, $s_i<0$. They can only consume the content jointly with the other user $j$ at this reduced price if user $j$ pays a sufficiently large subsidy $s_j$. Otherwise, user $i$ does not consume and does not have to pay.
High intermediate types $\theta_i \in (y,z]$ select a regular price $p_i\in(\underline p,\overline p]$ and no subsidy, $s_i=0$.
They can consume the content jointly with the other user $j$ if and only if this user selects a sufficiently large price $p_j$. Otherwise, user $i$ does not consume, but pays the price $p_i$ regardless. 
Finally, high-type users $\theta_i>z$ select the maximal price $\overline p$ and also a community subsidy, $s_i>0$. They always get to consume the good and have to pay $\overline p+s_i$ in any case. They consume jointly with the other user $j$ if the selected subsidy $s_i$ exceeds a threshold determined by the other user's choice $s_j.$

\textbf{``Altruistic" features:} 
To put it in a nutshell, our model can rationalize two elements of monetization schemes in the creator economy, voluntary payments (``donations") and community support features. 
At first glance, the efficacy of these features appears to be driven by generosity. However, both implementation features can work perfectly fine with purely self-interested agents when two important aspects of the creator economy are accounted for:
First, the non-rivalry of digital content with production costs entails that a user may want to pay more to increase the probability that the content is provided (or, alternatively, that the creator can continue her career rather than leaving the creator economy for a job in another industry). Second, a user can benefit from subsidizing other users to garner positive value network effects. Although we do not deny the importance (or existence) of altruism in small digital communities, we believe our model contributes to a better understanding of the full picture.

\textbf{Negative network effects$\to$ exclusivity bids:}
Some digital content is provided in a more exclusive fashion that is more akin to standard private-good provision. For example, a chess streamer may discuss games in front of a larger audience of subscribers but may also offer to give private and exclusive feedback on a subscriber's own game.\footnote{Note that our framework also allows to see private lessons and crowd lessons as two distinct products with different values (via $v$) and different costs (via $\varphi$).} As an alternative example, a live-streamer offering adult content may perform in front of a larger audience, but may change to a more private setting for additional payments. Figure \ref{subfig:implement_neg} illustrates how to employ such exclusivity bids in the 2-user environment with negative value network effects depicted in Panel \ref{subfig:smallNeg}.
Here, a user's action is to choose $(p_i,b_i)$, a price $p_i$ and an exclusivity bid $b_i$. Prices always have to be paid and divide users into three subscription tiers, while bids only have to be paid when they are successful in excluding the other user.
In the appendix, we jointly construct the outcome function and the equilibrium to implement the optimal allocation, where the types are again partitioned into the regions that are marked in the third row of Figure \ref{fig:MQ}.

Low-type users $\theta_i < x$  abstain by selecting $p_i=b_i=0$. 
Low intermediate types $\theta_i \in (x,y]$ select a flexible price $p_i \in (0,\overline p]$, but submit no exclusivity bid. They are ``non-subscribers" who can only consume the content in a group, which happens if the other user's exclusivity bid does not exceed a threshold. High intermediate types $\theta_i \in (y,z]$ are ``subscribers," who pay a fixed fee $p_i=\overline p$ and also submit an exclusivity bid $b_i>0$. While a sufficiently large bid can remove non-subscribers from the audience, they can also be kicked out of the audience by a ``premium subscriber" with a sufficiently large bid. Such users are types $\theta_i>z$, who pay the highest subscription fee $\overline{\overline{p}}$. They always consume the content and do so alone if their exclusivity bid is sufficiently high.
For $\pi$ so large that the good is essentially private, this procedure becomes a first-price auction with an entry fee.

\textbf{Large markets:}
Because we focus on how to finance the careers of smaller content creators, oligopsony is a central ingredient of this paper. In small markets, each individual user can have a strategic impact.
However, as discussed in the introduction, this industry  has also spawned superstars with very large audiences. A natural question is whether the optimal monetization of such creators' content is structurally different.
For example, many large video platforms impose an identical cost on every user to consume the cost by showing the same amount of ads to all users (and only those with a sufficiently large value watch the add to consume to content) rather than screening types. Such a monetization scheme can be optimal in large markets.

The next result shows that if positive network effects converge as the market (and audience) grows large, the optimal selling procedure becomes a simple posted price: all users whose value (in the large-audience limit) exceeds the price accept it to consume the content with certainty, while all others do not get the good and do not pay anything.
Here, not only the network externalities disappear due to the market size, but also the cost externality vanishes. In a large market, the fixed production cost becomes negligible and, consequently, so does a single user's impact on the production decision. 

\begin{lem} \label{lem:large}
Consider the large-market limit as $N\to \infty.$
Suppose $\varphi$ and $\psi(\theta_i,\cdot)$ are increasing and concave in $k$ for all $\theta_i$.
There is a cutoff $p$ such that $\lim_{N \to \infty } \overline q_i (\boldsymbol{\theta}) = \mathbbm{1}_{\theta_i \geq p}$ for all $\boldsymbol{\theta}$. That is, the optimal allocation is implementable by a simple posted price $p$.
\end{lem}

An alternative perspective on the result above is to consider it as the limit of the all-pay contribution implementation suggested before. Since the extractable revenue increases as the market grows, the fixed cost $c$ becomes essentially irrelevant for the production decision, and the network-effect adjusted marginal cost becomes zero. In the limit, the good is produced with probability one and the marginal network effects vanish. Hence, a user's extra contribution has no impact, and each user either selects the minimum price to be eligible for consumption or abstains.

When network externalities are negative, audience size $k^* = \argmax_k \{k \psi(\overline \theta,k)+\varphi(k)\}$ might be finite.
Hence, accepting all users above a threshold with a simple posted price is not optimal in general. As we have discussed before, such strong negative externalities move the model closer to a private-good setting. However, in a large-market limit, the optimal $k^*$-unit auction also trivially converges to a posted price $p=\overline \theta$, but there must be rationing.

\textbf{No cost, $c=0$:} The above implementation through simple price posting relies on two conditions: the cost externality and the value externality of users have to vanish. Figure \ref{fig:c0} illustrates the optimal allocation in examples when we set $c=0$. In all three panels, $\theta_i$ is uniformly distributed, and it is equal to the valuation when consuming alone. In the first panel, $\theta_i$ is also the value of shared consumption, while the value of joint consumption is multiplicatively inflated in the second panel and deflated in the third panel. It is easy to verify that only in the case without value network effects a simple posted price is optimal. In the other two cases, we can implement the optimal allocation with simpler versions of our discussions of community subsidies and exclusivity bids, respectively. 

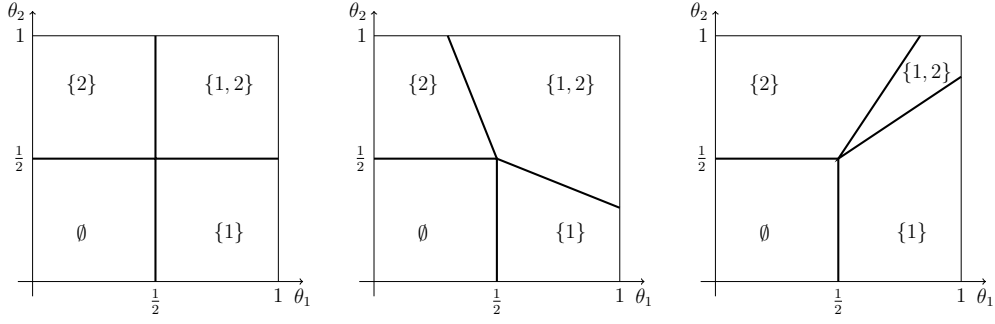
\begin{figure}[h!]
\begin{center}
\subfloat{\label{subfig:c0noNE}
\scalebox{0.65}{
\begin{tikzpicture}
\draw[->] (-0.3,0)  -- (5.5,0) node[below]{$\theta_1$};
\draw (5,0) node[below]{$1$};
\draw (0,5) node[left]{$1$};
\draw[->] (0,-0.3) -- (0,5.5) node[left] {$\theta_2$};
\draw[very thick] (0.5*5,0) node[below]{$\frac{1}{2}$} -- (0.5*5,0.5*5) -- (0,0.5*5)  node[left]{$\frac{1}{2}$};
\draw[very thick] (0.5*5,5) -- (0.5*5,0.5*5) -- (5,0.5*5);
\draw  (1, 1) node{$\emptyset$};
\draw  (1, 4) node{$\{ 2\}$};
\draw  (4, 1) node{$\{ 1\}$};
\draw  (4, 4) node{$\{ 1,2\}$};
\draw (0,5) -- (5,5) -- (5,0);
\end{tikzpicture}
}} 
\subfloat{ \label{subfig:c0posNE}
\scalebox{0.65}{
\begin{tikzpicture}
\draw[->] (-0.3,0)  -- (5.5,0) node[below]{$\theta_1$};
\draw (5,0) node[below]{$1$};
\draw (0,5) node[left]{$1$};
\draw[->] (0,-0.3) -- (0,5.5) node[left] {$\theta_2$};
\draw[very thick] (0.5*5,0) node[below]{$\frac{1}{2}$} -- (0.5*5,0.5*5) -- (0,0.5*5)  node[left]{$\frac{1}{2}$};
\draw[very thick] (0.3*5,5) -- (0.5*5,0.5*5) -- (5,0.3*5);
\draw  (1, 1) node{$\emptyset$};
\draw  (1, 4) node{$\{ 2\}$};
\draw  (4, 1) node{$\{ 1\}$};
\draw  (4, 4) node{$\{ 1,2\}$};
\draw (0,5) -- (5,5) -- (5,0);
\end{tikzpicture}
}}
\subfloat{ \label{subfig:c0negNE}
\scalebox{0.65}{
\begin{tikzpicture}
\draw[->] (-0.3,0)  -- (5.5,0) node[below]{$\theta_1$};
\draw (5,0) node[below]{$1$};
\draw (0,5) node[left]{$1$};
\draw[->] (0,-0.3) -- (0,5.5) node[left] {$\theta_2$};
\draw[very thick] (0.5*5,0) node[below]{$\frac{1}{2}$} -- (0.5*5,0.5*5) -- (0,0.5*5)  node[left]{$\frac{1}{2}$};
\draw[very thick] (5*5/6,5) -- (0.5*5,0.5*5) -- (5,5*5/6);
\draw  (1, 1) node{$\emptyset$};
\draw  (1, 4) node{$\{ 2\}$};
\draw  (4, 1) node{$\{ 1\}$};
\draw  (4.3, 4.25) node{$\{ 1,2\}$};
\draw (0,5) -- (5,5) -- (5,0);
\end{tikzpicture}
}}
 \caption{Examples of optimal allocations without (left), with positive (middle), and with negative (right) value network effects, when production costs are zero.} \label{fig:c0}
\end{center}
\end{figure}

\section{Conclusion} \label{sec:conclusion}

We find the profit-maximizing allocation of a seller offering a club good with network effects. The network effects only depend on the number of buyers that jointly consume the good, and we consider direct effects on the profit and effects on the buyers' valuations that can be positive or negative.
Revealing types truthfully is a dominant strategy in our direct mechanism, and we show that no higher profit can be obtained in any Bayesian Nash equilibrium of any other game.

Our model applies to a variety of settings, but our main application is the creator economy to which our indirect implementations are tailored.
We rationalize prominent features of monetization schemes in this industry. Specifically, we can explain donations to the creator or even other users. Such implementation features seem to be based on users' generosity, but the rationale that makes them optimal in our setting is that users take into account a cost externality and a value network externality. They are willing to pay more than others for the same good to increase the probability of provision, and they are willing to subsidize others when enlarging the audience is of inherent value to them. In contrast, when users dislike larger audiences, they are willing to pay extra to have others excluded.


\section*{Appendix}
\setcounter{subsection}{0}
\renewcommand{\thesubsection}{A.\Roman{subsection}}

\section{Proofs}

\begin{lem} \label{lem:sc-cond}
    Condition \eqref{eq:sc-cond} implies the value single-crossing condition for all $x>y, k\neq k'$, 
  \eq{ \label{eq:sc} \tag{SC}
\sign \{ (v(x,k) - v(x,k')) - (v(y,k) - v(y,k') ) \} = \sign \{     v(x,k) - v(x,k') \}.
    }
    Together with \eqref{eq:sc-cond2}, it implies the virtual-value single-crossing condition
     \eq{ \label{eq:sc-virt} \tag{SC-$\psi$}
     \begin{aligned}
         \sign \{ \psi(x,k) - \psi(x,k')) - (\psi(y,k) - \psi(y,k') \} = \\
         \sign \{     v(x,k) - v(x,k') \} \quad \forall x>y,k \neq k'.
     \end{aligned}
    }   
\end{lem}
\begin{proof}[Proof of Lemma \ref{lem:sc-cond}]
Fix any two levels $k,k'$, and suppose $v(x,k) > v(x,k')$.
To prove \eqref{eq:sc}, we show that $\xi(x)=v(x,k)-v(x,k')$ is increasing everywhere. Indeed,
$$\frac{\partial \xi}{\partial x} (x)= v_1(x,k) -  v_1(x,k')>0 \quad \mbox{by \eqref{eq:sc-cond}.}$$
To prove \eqref{eq:sc-virt}, we define $\widehat \xi(x)=\psi(x,k)-\psi(x,k')$. Let $\lambda(x)=(1-F(x))/f(x)\geq 0$, and let $\lambda'(x)\leq 0$ be its derivative. Indeed,
$$\frac{\partial \widehat \xi}{\partial x} (x)= (1 - \lambda'(x) ) \left(v_1(x,k) -  v_1(x,k') \right) - \lambda(x) \left( \frac{\partial v_1}{\partial \theta_i}(x,k) -\frac{\partial v_1}{\partial \theta_i}(x,k') \right) \geq 0 $$
by \eqref{eq:sc-cond} and \eqref{eq:sc-cond2}.
An analogous argument holds if $v(x,k) \leq v(x,k')$.
\end{proof}

\begin{proof}[Proof of Lemma \ref{lem:implementability}]
We first show that in any incentive-compatible mechanism \eqref{eq:cutoff}, \eqref{eq:q=1}, and \eqref{eq:integral-m} must hold.

 Fix any buyer $i$ and two types $x>y$. \eqref{eq:IC} requires that for each $\boldsymbol{\theta}_{-i}$ (replaced by $\cdot$ below)
 \ea{
 &U_i (x,\cdot) = v(x,k(x,\cdot)) q_i(x,\cdot) - m_i(x,\cdot) &\geq
 &v(x,k(y,\cdot)) q_i(y,\cdot) - m_i(y,\cdot) &\mbox{ and } \\
  &U_i (y,\cdot) = v(y,k(y,\cdot)) q_i(y,\cdot) - m_i(y,\cdot) &\geq
 &v(y,k(x,\cdot)) q_i(x,\cdot) - m_i(x,\cdot).
 }
  Subtracting the inequalities yields
 \ea{
q_i(x,\cdot) \big(  v(x,k(x,\cdot)) -v(y,k(x,\cdot))\big) \geq 
q_i(y,\cdot) \big( v(x,k(y,\cdot)) -v(y,k(y,\cdot)) \big).
 }
 Hence, for all $x,y,\boldsymbol{\theta}_{-i}$ with 
 $q_i(x,\boldsymbol{\theta}_{-i}) =q(y,\boldsymbol{\theta}_{-i})=1$, we obtain \eqref{eq:q=1}.
 Because $x>y$ and $v$ is increasing in $\theta_i$ for all $k$, we obtain for all other $x,y,\boldsymbol{\theta}_{-i}$ that $q_i$ must be weakly increasing in $\theta_i$,
 \ea{
q_i(x,\cdot)  \geq 
q_i(y,\cdot) \left( \frac{v(x,k(y,\cdot)) -v(y,k(y,\cdot)) }{  v(x,k(x,\cdot)) -v(y,k(x,\cdot))}  \right).
 }
 Because the second factor on the right-hand side is positive, it must be that $q_i(x,\cdot)=1$ when $q_i(y,\cdot)=1$, and it must be that
 $q_i(y,\cdot)=0$ when $q_i(x,\cdot)=0$. Therefore, for all $\boldsymbol{\theta}_{-i}$, $q_i$ is either constant or has exactly one jump upwards, i.e., \eqref{eq:cutoff} holds. 
 
Moreover, considering $x=y+\delta$ with $q_i(x,\cdot)=q_i(y,\cdot)=1$, we obtain
\ea{
 \lim_{\delta \to 0} \frac{U_i(y+\delta,\cdot) - U(y,\cdot)}{\delta} &\geq \lim_{\delta \to 0} \frac{1 \big( v(y+\delta,k(y+\delta,\cdot)) -v(y,k(y,\cdot)) \big)}{\delta} \\
 &= 
 v_1(y,k(y,\cdot)),\\
  \lim_{\delta \to 0} \frac{U_i(x,\cdot) - U(x-\delta,\cdot)}{\delta} &\leq \lim_{\delta \to 0} \frac{1 \big( v(x,k(x,\cdot)) -v(x-\delta,k(x-\delta,\cdot)) \big)}{\delta} \\
 &= 
 v_1(x,k(x,\cdot)),
}
which implies that $U'_i(x,\boldsymbol{\theta}_{-i})= v_1(x,k(x))$ wherever $q_i$ is equal to one. Since \eqref{eq:IC} also implies that $U_i$ is Lipschitz-continuous, it is differentiable almost everywhere and equals the integral over its derivative. Hence, \eqref{eq:integral-m} follows from rearranging.

Now, suppose \eqref{eq:cutoff}, \eqref{eq:q=1}, and \eqref{eq:integral-m} hold. We can rewrite \eqref{eq:IC} as 
\ea{
U_i(x,\cdot) &\geq u_i(y,\cdot|x)=q_i(y,\cdot) v(x,k(y,\cdot) - m_i(y,\cdot)\\
U_i(x,\cdot) &\geq u_i(y,\cdot|x) + q_i(y,\cdot) v(y,k(y,\cdot) - q_i(y,\cdot) v(y,k(y,\cdot)\\
U_i(x,\cdot) &\geq q_i(y,\cdot) \big( v(x,k(y,\cdot) - v(y,k(y,\cdot) \big) + U_i(y,\cdot) \\
\int_y^x q_i(t,\cdot) v_1(t,k(t,\cdot)) dt &\geq 
\int_y^x q_i(y,\cdot) v_1(t,k(y,\cdot))dt,
}
which is implied by \eqref{eq:cutoff} and \eqref{eq:q=1}.
\end{proof}

\begin{proof}[Proof of Lemma \ref{lem:sc}]
By \eqref{eq:q=1} of Lemma \ref{lem:implementability}, \eqref{eq:IC} requires that for all $x>y >\underline x(\boldsymbol{\theta}_{-i})$,
\eq{ \label{eq:1}
\big( v(x,k(x,\cdot) ) -v(y,k(x,\cdot)) \big) - \big( v(x,k(y,\cdot)) -v(y,k(y,\cdot)) \big)
\geq  0.
}
By contradiction suppose that $k(x,\cdot)=k_j$ and $k(y,\cdot)=k_{j'}$ with $j<j'$, i.e., the lower type $y$ receives a ``better" consumer set size. 
We arrive at a contradiction with \eqref{eq:sc} of Lemma \ref{lem:sc-cond},
\eq{ \label{eq:2}
\big( v(x,k(x,\cdot) ) -v(y,k(x,\cdot)) \big) - \big( v(x,k(y,\cdot)) -v(y,k(y,\cdot)) \big) <0. }
Vice versa, suppose that for all $x>y$, $k(x,\cdot)=k_j$ and $k(y,\cdot)=k_{j'}$ with $j \geq j'$. The negation of \eqref{eq:2} implies \eqref{eq:1}. 
\end{proof}

\begin{proof}[Proof of Lemma \ref{lem:aggr-provision}]
Fix type vector $\boldsymbol\theta$, and let $J^*$ be the optimal consumer set in the relaxed problem. That is, $q_i (\boldsymbol \theta)= \mathbbm{1}_{i \in  J^*}$.

Contradicting the lemma, suppose $J^* \neq \varnothing$
although \eqref{eq:cost-covered} is violated for all $J$. The seller makes a loss that can be avoided by setting $q_i (\boldsymbol \theta)=0$ for all $i$.
Analogously, $J^* = \varnothing$ cannot be optimal if a profitable $J$ satisfying \eqref{eq:cost-covered} exists.
\end{proof}

\begin{proof}[Proof of Lemma \ref{lem:aggr-structure}]
The right-hand side of \eqref{eq:J>J'} is zero when $|J|=|J'|$, while the left-hand side is positive when $J\neq J'$ has structure \eqref{eq:agg-structure}. 
\end{proof}

\begin{proof}[Proof of Lemma \ref{lem:finder}]
Fix any $\boldsymbol \theta$.
In keeping with Lemma \ref{lem:aggr-provision}, the constraint in \eqref{eq:reduced-max} ensures that the good is only provided when provision is profitable.
By Lemma \ref{lem:aggr-structure}, the optimal allocation rule is a simple exclusion policy, and the solution to \eqref{eq:reduced-max} identifies the most profitable such allocation rule.
\end{proof}

\begin{proof}[Proof of Proposition \ref{prop:optimal}]
To show that $\overline q$ is incentive-compatible, we verify that it always satisfies the three conditions of Lemma \ref{lem:implementability}, i.e., \eqref{eq:cutoff}, \eqref{eq:q=1}, and \eqref{eq:integral-m}. The allocation is individually rational if it is incentive-compatible and the participation constraint binds for type $\theta_i=0$, which is true.

\eqref{eq:integral-m} is satisfied by construction \eqref{eq:integral-m2}.

We fix a buyer $i$ and consider two possible types $x>y$ for him. We also fix the types of all other buyers $\boldsymbol{\theta}_{-i}$ and order them, $\theta_j>\theta_{j+1}$ for all $j$.

To show \eqref{eq:cutoff}, we need that 
 \ea{
    \overline q_i(y,\boldsymbol{\theta}_{-i})=1  \implies \quad  q_i(x,\boldsymbol{\theta}_{-i})=1,  \mbox{ and } \quad
    \overline q_i(x,\boldsymbol{\theta}_{-i})=0  \implies \quad  q_i(y,\boldsymbol{\theta}_{-i})=0.
    }   
The first part cannot be violated because if the maximization of \eqref{eq:reduced-max} leads to an inclusion of type $y$ and $k-1$ others (for any $k$), the same is true for type $x$. That is, if it holds that
\ea{
\psi(y,k) + \sum_{j=1}^{k-1} \psi(\theta_j,k) + \varphi(k)   \geq  \sum_{j=1}^{k'} \psi(\theta_j,k') + \varphi (k') 
}
for any $k'$, regularity implies this condition also holds for type $x$. Similarly, if the condition does not hold for type $x$, it cannot hold for type $y$.
Also because of regularity, we have that if there exists a consumer set including type $y$ for which the provision condition \eqref{eq:cost-covered} holds, the condition also holds for the same consumer set including type $x>y$. Therefore, given any $\boldsymbol \theta_{-i}$, $\overline q_i$ is weakly increasing in type $\theta_i$ for all $i$ such that \eqref{eq:cutoff} holds.

Suppose that $\overline q_i(x,\cdot)=\overline q_i(y,\cdot)=1$, and $\overline k(x,\cdot)=k_x \neq \overline k(y,\cdot)=k_y$ and both satisfy the cost constraint. By construction, it must be that
\eq{ \label{eq:psi-optimum}
\begin{aligned}
&\psi(x , k_x) + \sum_{j=1,j\neq i}^{k_x} \psi (\theta_j, k_x) + \varphi(k_x) 
&\geq 
\psi(x , k_y) + \sum_{j=1,j\neq i}^{k_y} \psi (\theta_j, k_y) + \varphi(k_y),  \mbox{ and} \\
&\psi(y , k_y) + \sum_{j=1,j\neq i}^{k_y} \psi (\theta_j, k_y) + \varphi(k_y) 
&\geq 
\psi(y , k_x) + \sum_{j=1,j\neq i}^{k_x} \psi (\theta_j,k_x) + \varphi(k_x).
\end{aligned}
}
Subtracting these two inequalities, we obtain
\ea{ 
\psi(x , k_x) - \psi(y , k_x) \geq \psi(x , k_y) - \psi(y , k_y), 
}
which is \eqref{eq:sc-virt} of Lemma \ref{lem:sc-cond} with $v(x,k_x)\geq v(x,k_y)$.
Thus, \eqref{eq:q=1} holds.

Hence, $\langle \overline q, \overline m \rangle$ is incentive compatible and individually rational, and thus also solves the full constrained problem.
\end{proof}

\begin{proof}[Proof of Lemma \ref{lem:complements}]
Consider buyer $i$ his types $x>y$, and let $\boldsymbol{\theta}_x=(x,\boldsymbol{\theta}_{-i})$ and     $\boldsymbol{\theta}_y= (y,\boldsymbol{\theta}_{-i})$ with $\overline k(\boldsymbol{\theta}_x)=k_x$ and $\overline k(\boldsymbol{\theta}_y)=k_y$.

Suppose that $i$ is in the consumer set when $\theta_i=y$. By incentive compatibility, he is also in the consumer set when $\theta_i=x>y.$
Ignoring the cost constraint of \eqref{eq:reduced-max}, the argument around \eqref{eq:psi-optimum} in the proof of Proposition \ref{prop:optimal} above confirms that $k_x \geq k_y \geq i$. That is, the seller weakly wants to expand the consumer set because of the type increase when value network effects are positive. Since the cost $C(k)$ is decreasing in $k$, this larger set is also feasible.    

Now, consider some other buyer, who has the $\ell$-th highest type in $\boldsymbol{\theta}_x$ and either also the $\ell$-th or the $(\ell-1)$-th highest type in $\boldsymbol{\theta}_y$. This buyer cannot be kicked out of the consumer set due to an increase in type $\theta_i$, because
$\ell-1<\ell \leq k_y \leq k_x$, and he also cannot be added to the consumer set due to a decrease in type $\theta_i$, because $\ell> \ell-1>k_x\geq k_y$.
\end{proof}

\begin{proof}[Proof of Lemma \ref{lem:Bayes}]
Employing the envelope theorem \citep{milgrom2002envelope}, we can rewrite a buyer's expected utility \eqref{eq:interim-utility} in any implementable mechanism as
\ea{
U(\theta_i) = U(0) + \int_0^{\theta_i}\sum_{k=1}^N Q_i^k(t) v_1(t,k) d t
}
such that given $U(0)=0$ the expected payment is
\eq{
  M_i (\theta_i) =  \sum_{k=1}^N Q_i^k(\theta_i) v(\theta_i,k) - \int_0^{\theta_i}\sum_{k=1}^N Q_i^k(t) v_1(t,k) d t
 }
such that
\ea{
\begin{aligned}
  \mathbb E [M_i (\theta_i) ] 
  &=  \int_0^{\overline \theta} \sum_{k=1}^N Q_i^k(\theta_i) \psi (\theta_i,k) f(\theta_i) d\theta_i \\
    &=  \int_{\Theta^N} q_i(\theta_i,\boldsymbol{\theta}_{-i}) \psi (\theta_i,k(\theta_i,\boldsymbol{\theta}_{-i}))  d G(\theta_i,\boldsymbol{\theta}_{-i}) \\
 \end{aligned}
}
so that the seller's expected profit is again \eqref{eq:virt-value-profit}.
Hence, the solution to the relaxed problem solves the fully constrained problem under \eqref{eq:B-IC} if the solution $\langle \overline{q},\overline{m}\rangle$ satisfies all constraints. Since it satisfies the stronger \eqref{eq:IC}, also \eqref{eq:B-IC} is satisfied.
\end{proof}

\begin{proof}[Proof of Lemma \ref{lem:always-never}]
Part a) follows from Lemma \ref{lem:aggr-provision}.
Suppose \eqref{eq:never-provide} holds. For any $k\leq N$, a set $J_k=\{i:i\leq k\}$ is the consumer set following from the corresponding simple exclusion policy, and $\Psi(\boldsymbol{\overline \theta}_{\leq k} |J_k) < C (k)$. By regularity, this is also true for all other type vectors.
Suppose \eqref{eq:never-provide} does not hold for some $k$. Then there exists a type vector realization such that  $\Psi(\boldsymbol{\theta}_{\leq k} |J_k)\geq  C (k)$.

Suppose \eqref{eq:always-provide} holds, and suppose all types are zero, $\boldsymbol{\theta}=\boldsymbol{0}$. Part (i) implies that accepting all $N$ buyers covers the adjusted total cost, and part (ii) implies that removing any buyer reduces profits. Hence $\overline q_i(\boldsymbol{0})=1$ for all $i$. The incentive compatibility of $\langle \overline{q},\overline{m}\rangle$---Lemma \ref{lem:implementability}(i)---implies the consumer set cannot be smaller for any other type vector. Individual rationality and $v(0,k)=0$ for all $k$ imply $\overline m_i(\boldsymbol{0})=0$ for all $i$.

Suppose one part of \eqref{eq:always-provide} does not hold.
Either accepting all buyers does not cover the adjusted cost or excluding some buyers for some type vector increases profit. Hence, accepting all buyers is not optimal. Combined with the previous paragraph, statement b) follows.
\end{proof}

\begin{proof}[Proof of Lemma \ref{lem:M-incr}]
From the integral form \eqref{eq:integral-m}, we can see that
\ea{ 
   m_i(\theta_i,\boldsymbol{\theta}_{-i}) = v(\theta_i,k(\theta_i, \boldsymbol{\theta}_{-i})) - \int_{\underline x (\boldsymbol{\theta}_{-i})}^{\theta_i} v_1(t,k(t,\boldsymbol{\theta}_{-i})) d t,
}
is weakly increasing in $\theta_i$ for any $\boldsymbol{\theta}_{-i}$. Hence, it is weakly increasing in expectation. The equation also implies that any two types who have the same payment for all $\boldsymbol{\theta}_{-i}$ must be indifferent over their contracts.
\end{proof}

\begin{proof}[Proof of Proposition \ref{prop:allpay}]
Take any setting and the corresponding optimal DRM $\langle \overline q, \overline m\rangle$. We construct the corresponding all-pay contribution mechanism jointly with the implementing Bayesian Nash equilibrium, $\langle P,g^*, \rho^* \rangle$.

First, set $g^*_i(p_i,\boldsymbol{p}_{-i}) = (0,p_i)$ for all $\boldsymbol{p}_{-i}$ with $p_i \not\in [\overline M_i (\underline y), \overline M_i(\overline y)]$  such that choosing any such price is weakly dominated by selecting price $p_i=0$. Suppose all types $\theta_i< \underline y$ select $\rho^*_i(\theta_i)=0$.

Second, suppose any type $\theta_i \in [\underline y,\overline y)$ selects price 
\ea{
\rho^*_i(\theta_i) = \overline M_i ( \theta_i),}
and any type $\theta_i \geq \overline y$ selects price $\rho^*_i(\theta_i) = \overline M_i ( \overline \theta)$.

Third, since $\overline M_i$ is invertible, we can construct 
\ea{
\widetilde g_i(p_i,\mathbf{p}_{-i}) = (\overline q_i ( \overline M_i^{-1} (p_i),\overline{\mathbf{M}}_{-i}^{-1} (\mathbf{p}_{-i})),p_i)
}
such that under the supposed strategy profile $g_i^*(p_i,\boldsymbol{p}_{-i}) = (\overline q_i (\boldsymbol{\theta}),M_i(\theta_i))$ for all $i$ and all $\boldsymbol{\theta}$. If for several types $x\neq y$ we have $\overline M_i(x)=\overline M_i(y)$, Lemma \ref{lem:M-incr} ensures that $q_i(x,\boldsymbol{\theta}_{-i})=q_i(y,\boldsymbol{\theta}_{-i})$ and $k(x,\boldsymbol{\theta}_{-i})=k(y,\boldsymbol{\theta}_{-i})$ such that it does not matter which types is used as input among types in a flat region of $\overline M_i$.

Under the proposed strategy profile, buyer $i$'s expected utility from following strategy $\rho^*_i$ is
\ea{
 u_i (\rho^*_i(\theta_i),\theta_i) &= \mathbb{E}_{\boldsymbol{\theta}_{-i}} \left[ \widetilde q_i \big(\rho^*_{i}(\theta_i),{\rho^*}_{-i}(\boldsymbol{\theta}_{-i}) \big) v \big(\theta_i, \widetilde k(\rho^*_{i}(\theta_i),{\rho^*}_{-i}(\boldsymbol{\theta}_{-i}) \big) - \rho^*_i(\theta_i))
\right] \\
&= \mathbb{E}_{\boldsymbol{\theta}_{-i}} \left[ \overline q_i (\theta_{i},\boldsymbol{\theta}_{-i}) v (\theta_i, \overline k(\theta_i, \boldsymbol{\theta}_{-i})) - \overline m_i(\theta_{i},\boldsymbol{\theta}_{-i}))
\right],
}
which by incentive compatibility and individual rationality of $\langle \overline q, \overline m\rangle$ is at least as large as the expected utility from any deviation to any price $p_i\in \{0\} \cup [\overline M_i (\underline y), \overline M_i(\overline y)]$, which dominate all other prices. Hence, ${\rho^*}$ is a Bayesian Nash equilibrium of the constructed all-pay contribution mechanism, and it implements the optimal allocation for all $\boldsymbol{\theta}$.
\end{proof}

\begin{proof}[Proof of Lemma \ref{lem:large}]
Because the functions are concave and increasing in $k$, they converge. Let 
\ea{
\lim\limits_{k \to \infty} \varphi (k) = \varphi^* \quad \mbox{ and } \quad
\lim\limits_{k \to \infty} v(\theta_i,k) = v^*(\theta_i).
}

Consider an increasing sequence of market sizes $N^{(n)}=(1,2,\dots)$.
Let $\boldsymbol{\theta}^{(n)}$ be the sequence of type profile realizations with $\boldsymbol{\theta}^n = (\boldsymbol{\theta}^{n-1},\theta_n^n)$ (before ordering types according to virtual values), i.e., one more type draw is added each step along the sequence and the other types remain fixed $\theta^n_i=\theta^{n'}_i$ for all $i \leq n<n'$. Let $\overline k^n(\boldsymbol{\theta}^n) = |\overline J (\boldsymbol{\theta}^{n})|$ be the optimal consumer set size at market size $n$ along the sequence.
Let  $\psi^* (\theta_i)= \lim_{k\to \infty} \psi (\theta_i,k) = v^*(\theta_i) - \frac{1-F(\theta_i)}{f(\theta_i)} \frac{\partial v^*(\theta_i) }{\partial \theta_i}$.

$\overline k^n(\boldsymbol{\theta}^n)$ is weakly increasing in $n$ for all type profiles, because the left-hand side of the inequality below  (after ordering types) is  increasing in $n$ and $j$, while the right-hand side converges to a constant as $n$ and $j$ increase (and the production cost $c$ is also constant)
\ea{
&\sum_{i=\overline k^n}^{\overline k^n+j} \psi(\theta_i^n,\overline k^n+j)
\geq &\gamma (\overline k^n, \overline k^n+j, \boldsymbol{\theta}_{\leq \overline k^n}) \\
&\gamma (\overline k^n, \overline k^n+j, \boldsymbol{\theta}_{\leq \overline k^n}) &= \underbrace{\varphi(\overline k^n)- \varphi( \overline k^n+j)}_{\to \varphi(\overline k^n)- \varphi^*} +  \sum_{i= 1}^{ \overline k^n} \underbrace{[\psi (\theta_i, \overline k^n)-\psi (\theta_i, \overline k^n+j)] }_{\to \psi (\overline \theta, \overline k^n)- \psi^* (\overline \theta)}.
}
That is, as the market grows it becomes weakly profitable to increase the consumer set. Finally, for all $j$
\ea{
\begin{aligned}
  \lim_{k\to \infty}  \gamma (k,k+j ,\boldsymbol{\theta}_{\leq k}) =  \overbrace{\varphi(k)- \varphi(k+j)}^{\to 0} +  \sum_{i= 1}^{k} \overbrace{[\psi (\theta_i,k)-\psi (\theta_i,k+j)] }^{\to 0} = 0.
\end{aligned}
}
That is, as the consumer set grows the cutoff for the virtual value to be admitted to the consumer set converges to zero. Consequently, in the large-market limit, the optimal allocation $\overline q^*_i(\boldsymbol{\theta}) = \mathbbm{1}_{\psi^*(\theta_i) \geq 0}$ can be implemented by posting a price $p$ such that $\psi^*(p)=0.$
\end{proof}


\section{Implementations: Details behind Figure \ref{fig:implementation}}

Here, we construct the games and equilibria that implement the respective optimal allocation as described in Figure \ref{fig:implementation}. We use specific settings, but the general idea of our implementations generalizes. Our working example, is a 2-user setting with the valuation function \eqref{eq:v-pi} used in Figure \ref{fig:benchmarks},
\ea{ 
v(\theta_i,k)=\begin{cases}
			\pi \theta_i & \text{if }   k=1, \\
			(1-\pi)\theta_i & \text{if } k=2,
		 \end{cases}
}
and we assume no profit network effects, $\varphi(k)=0$ for all $k$. Types $\theta_i$ are independent draws from a uniform distribution, $F(\theta_i)=\theta_i$ for  $\theta_i \in [0,1]$. The production cost is $c= \nicefrac{1}{4}$. Both games constructed below are symmetric deterministic simultaneous move games that are only played by the two users, and we find a pure-strategy Bayesian Nash equilibrium.

\subsubsection*{Positive value network effects and community subsidies} 

Figure \ref{fig:pos-again} shows the optimal allocation for our illustrative setting with positive value network effects, $\pi=\nicefrac{3}{8}$. Notice that the general structure of the optimal allocation does not hinge on the specifics of our example. First, adding profit network effects has no qualitative impact, see, e.g., Panel \ref{subfig:NE} in the main text. Second, the impact of changing the type distribution would simply be that the curves connecting points $(x,1)$ and  $(y,z)$, and $(1,x)$ and $(y,z)$ would not be linear, but they would still be monotone and decreasing. Third, changing the cost $c$ also has no qualitative impact as long as all possible consumer sets are optimal for some type combination. Finally, the construction below is valid for all $\pi \in (\frac{1}{4},\frac{1}{2})$. Outside of this interval, either there is only joint consumption or the network effects are negative. There is a plethora of valuation functions outside of the structure of \eqref{eq:v-pi} that lead to a similar structure of the optimal allocation.

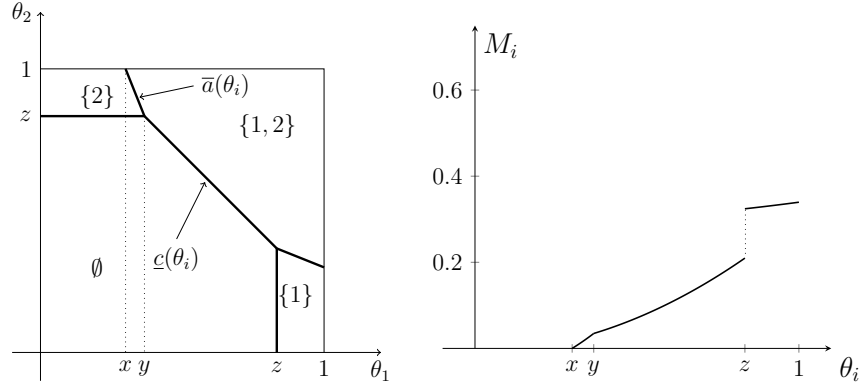
\begin{figure}[!htp]
\begin{center}
\subfloat{ \label{subfig:smallPos2} 
\scalebox{0.75}{
\begin{tikzpicture}
\draw[->] (-0.5,0)  -- (6,0) node[below]{$\theta_1$};
\draw (5,0) node[below]{$1$};
\draw (0,5) node[left]{$1$};
\draw[->] (0,-0.5) -- (0,6) node[left] {$\theta_2$};
\draw[very thick] (5/6*5,0) node[below]{$z$}-- (5/6*5,11/30*5) -- (5,3/10*5);
\draw[very thick] (0,5/6*5) node[left]{$z$} -- (11/30*5,5/6*5) -- (3/10*5,5);
\draw[very thick] (5/6*5,11/30*5) -- (11/30*5,5/6*5);
\draw[dotted] (11/30*5,5/6*5) -- (11/30*5,0) node[below]{$y$}; 
\draw[dotted] (3/10*5,5) -- (3/10*5,0) node[below]{$x$}; 
\draw  (4.5, 1) node{$\{1\}$};
\draw  (1, 4.5) node{$\{2\}$};
\draw  (1, 1.5) node{$\emptyset$};
\draw  (4, 4) node{$\{ 1,2 \}$};
\draw (0,5) -- (5,5) -- (5,0);
\draw[->] (2.7,4.7) node[right]{$\overline a(\theta_i)$}  -- (1.75,4.5);
    \draw[->] (2.4,2) node[below]{$\underline c(\theta_i)$}  -- (2.9,3);
\end{tikzpicture}
}}
\subfloat{ 
\scalebox{0.75}{
\begin{tikzpicture}[
declare function={
    Ma(\x)= ((1/128)*(100*\x*\x-9);
   Ma2(\x)= (5*\x*\x)/16-7/960;
   Ma3(\x)= (1/480)*(24*(\x*\x)+139);
  },
  ]
\begin{axis}[%
            domain=0:1,
            yscale=1,
            axis lines=middle,
            xlabel={\large $\theta_i$},
            ylabel={\large $M_i$},
            ymin=0,
            ymax=0.75,
            xmin=-0.1,
            xmax=1.1,
            samples=166,
            xlabel style={below right},
            xtick={0,3/10,11/30,5/6,1},
    xticklabels={0,$x$,$y$,$z$,1},
        ]
     \addplot [thick,domain=3/10:11/30-0.00001] {Ma(x)};
     \addplot [thick,domain=11/30:5/6-0.00001] {Ma2(x)};
     \addplot [thick,domain=5/6:1] {Ma3(x)};
     \draw [dotted] (93.5,220) -- (93.5,325);
    \end{axis}
\end{tikzpicture}
}}
 \caption{Optimal allocation (left) and optimal interim expected transfer $M_i(\theta_i)$ (right) for $\pi=\nicefrac{3}{8}$, and $\theta_i \sim U[0,1]$.}\label{fig:pos-again}
\end{center}
\end{figure}

In the example, $x=\nicefrac{3}{10}, y=\nicefrac{11}{30}$, and $z=\nicefrac{5}{6}$, but we can find $x<y<z$ in any other setting fitting the description above. We define
\eq{ \label{eq:xyzPos} 
\begin{aligned}
&x: &\psi(x,2)+\psi(\overline \theta,2)=c, \\
&y: &\psi(y,2)+\psi(z,2)=c, \\
&z: &\psi(z,1)=c.
\end{aligned}
}
Let us first describe the optimal allocation that our constructed equilibrium and game shall replicate.
Types $\theta_i \leq x$ never get the good and never pay anything. 
Types $\theta_i \in (x,y]$ never consume the good alone, and they consume the good jointly if type $\theta_j\geq z$ is large enough,
\eq{ \label{eq:alloc-xyPOS}
\psi(\theta_i,2)+\psi(\theta_j,2)>\max \{\psi(\theta_j,1),c\}=\psi(\theta_j,1)
\iff \theta_j > \overline a(\theta_i).
}
Types $\theta_i \in (y,z]$ also never consume the good alone, 
and they consume the good jointly if type $\theta_j$ is large enough,
\eq{ \label{eq:alloc-yzPOS}
\psi(\theta_i,2)+\psi(\theta_j,2) > c
\iff \theta_j > \underline c(\theta_i).
}
Finally, types $\theta_i>z$ always get the good: alone if $\theta_j < \overline a(\theta_i)$, and jointly otherwise. 

Now, we jointly construct the equilibrium and the game described in Figure \ref{subfig:implement_pos}.
 First, let actions be a pair $\alpha_i=(p_i,s_i) \in \mathbb R^2$. Next, define 
 \eq{ \label{eq:alloc-z1POS}
 \underline p= \overline M_i(y), \quad \mbox{ and } \quad \overline p= \overline M_i(z).
 }

\textbf{Strategy:}
Consider the following strategy $\sigma^*_i$:
\begin{enumerate}
    \item Types $\theta_i\in [0,x]$ select $(p_i,s_i)=(0,0).$
    \item Types $\theta_i\in (x,y]$ select $p_i=\underline p$ and $s_i= s^{x,y}(\theta_i)$, where 
    \ea{
    s^{x,y}: \quad \overline Q_i(\theta_i) (\underline p + s^{x,y}(\theta_i)) = \overline M_i(\theta_i).}
    \item Types $\theta_i\in (y,z]$ select $s_i=0$ and $p_i= p^{y,z}(\theta_i)$, where
    \ea{
    p^{y,z}: \quad p^{y,z}(\theta_i) = \overline M_i(\theta_i).}
    \item Types $\theta_i\in (z,1]$ select $s_i=s^{z,1}(\theta_i)$ and $p_i= \overline p$, where
    \ea{
    s^{z,1}: \quad: s^{z,1}(\theta_i) + \overline p = \overline M_i(\theta_i).}
    \end{enumerate}
    First, $s^{x,y}$ is strictly increasing and negative over its domain $[x,y]$, because $\overline M_i$ is strictly increasing and below $\underline p$ there. 
    Second, $p^{y,z}$ has range $[\underline p,\overline p]$ and is strictly increasing  over its domain $[y,z]$, because it is equal to $\overline M_i$ there.
    Third, $s^{z,1}$ is strictly increasing and positive over its domain $[z,1]$, because $\overline M_i$ is strictly increasing and above $\overline p$ there. 

\textbf{Outcome function:} Our goal is to construct an outcome function $g^*$  such that 
(i) it implements the optimal allocation in conjunction with a symmetric strategy profile $\sigma^*=(\sigma_1^*,\sigma_2^*)$, and (ii) $\sigma^*$ is a Bayesian Nash equilibrium. To this end, it is easiest to first render all actions inconsistent with $\sigma_i^*$ irrelevant. That is, define $g^*$ such that for all $(p_j,s_j)$
\eq{ \label{eq:oustide-of-equiPOS}
\begin{aligned}
   g^*_i \big( (p_i,s_i),(p_j,s_j) \big) = (0,p_i) \quad \forall (p_i,s_i) \mbox{ such that }\\
   p_i \not\in[\underline p, \overline p] \mbox{ or } s_i \not\in [-\underline p,\overline M_i(\overline \theta)-\overline p]
   \mbox{ or } p_i \in(\underline p, \overline p) \wedge s_i \neq 0\\
   \mbox{ or } p_i=\underline p \wedge s_i \not\in [-\underline p,0]
   \mbox{ or } p_i=\overline p \wedge s_i \not\in [0,\overline M_i(\overline \theta)-\overline p].
\end{aligned}
}
For actions consistent with $\sigma^*$, we ensure that types $\theta_i\in (x,y]$ playing $\sigma_i^*$ get their designated expected outcome and payoff from the optimal mechanism $\langle \overline q,\overline m \rangle$ for all $\theta_j$ playing $\sigma^*_j$. For all $(p_i,s_i)$ with $p_i= \underline p$ and $s_i \in [-\underline p,0]$
set
\eq{ \label{eq:equi-xyPOS}
\begin{aligned}
   g^*_i \big( (p_i,s_i),(p_j,s_j) \big) = 
   \begin{cases}
			(1, \underline p + s_i)  & \text{if }   s_j \geq s(s_i), \\
			(0,0) & \text{ otherwise},
		 \end{cases}\\
   \mbox{where } s (s_i) = s^{z,1} ( \overline a ( (s^{x,y})^{-1} (s_i)))
\end{aligned}
}
such that given $\sigma^*$ and $\theta_i \in [x,y]$, $i$ gets the good if and only if 
 $\theta_j \in [z,1]$ and
\ea{
 s^{z,1}(\theta_j) \geq s ( s^{x,y}(\theta_i) ) \iff
s^{z,1}(\theta_j) \geq  s^{z,1} ( \overline a ( (s^{x,y})^{-1} (s^{x,y}(\theta_i))
\iff \theta_j \geq \overline a(\theta_i),
}
which is \eqref{eq:alloc-xyPOS}. By construction of $\sigma^*$, the interim expected payoff of types $\theta_i\in (x,y]$ is the same as in the optimal mechanism.

Next, we  consider types $\theta_i\in [y,z)$. For all $(p_i,s_i)$ with $s_i= 0$ and $p_i \in (\underline p, \overline p)$ set
\eq{ \label{eq:equi-yzPOS}
\begin{aligned}
   g^*_i \big( (p_i,s_i),(p_j,s_j) \big) = 
   \begin{cases}
			(1, p_i)  & \text{if }   p_j \geq p(p_i), \\
			(0,p_i) & \text{ otherwise},
		 \end{cases}\\
   \mbox{where } p (s_i) = p^{y,z} ( \underline c( (p^{y,z})^{-1} (p_i)))
\end{aligned}
}
such that given $\sigma^*$ and $\theta_i \in [x,y]$, $i$ gets the good if and only if either
 $\theta_j \in [z,1]$ (because $\overline p>p(p_i)$ for all relevant $p_i$) or  
 $\theta_j \in [x,y]$ and
\ea{
 p^{y,z}(\theta_j) \geq p ( p^{y,z}(\theta_i) ) \iff
p^{y,z}(\theta_j) \geq  p^{y,z} ( \underline c ( (p^{y,z})^{-1} (p^{y,z}(\theta_i))
\iff \theta_j \geq \underline c(\theta_i),
}
which is \eqref{eq:alloc-yzPOS}. By construction of $\sigma^*$, the interim expected payoff of types $\theta_i\in (y,z]$ is the same as in the optimal mechanism.
All of the above is then also true for types $\theta_i >z$. By \eqref{eq:oustide-of-equiPOS}, types $\theta_i <x$ are clearly (weakly) best off by choosing $p_i=s_i=0$ and, hence, also get their designated outcomes and payoff.

Finally, $\sigma^*$ must be a Bayesian Nash equilibrium because the original optimal mechanism that we replicated is Bayesian incentive compatible. That is, no type $\theta_i$ has an incentive to deviate to another type's action given the other user plays $\sigma_j(\theta_j)$. Moreover, deviating to any $(p_i,s_i)$-combination outside of $\sigma_i^*$ is not profitable on account of \eqref{eq:oustide-of-equiPOS}.

\subsubsection*{Negative value network effects and exclusivity bids} 

Figure \ref{fig:neg-again} shows the optimal allocation for our illustrative setting with negative value network effects, $\pi=\nicefrac{5}{8}$. Again, the general idea of the implementation with exclusivity bids does not hinge on the specifics of our example.

\begin{figure}[!htp]
\begin{center}
\subfloat{  
\scalebox{0.75}{
        \begin{tikzpicture}
        \draw[->] (-0.5,0)  -- (6,0) node[below]{$\theta_1$};
        \draw (5,0) node[below]{$1$};
        \draw (0,5) node[left]{$1$};
        \draw[->] (0,-0.5) -- (0,6) node[left] {$\theta_2$};
        \draw[very thick] (7/10*5,0) -- (7/10*5,19/30*5) -- (5,5*5/6);
        \draw[very thick] (0,7/10*5) -- (19/30*5,7/10*5) -- (5*5/6,5);
        \draw[very thick] (7/10*5,19/30*5) -- (19/30*5,7/10*5);
        
        \draw  (4.5, 1) node{$\{1\}$};
        \draw  (1, 4.5) node{$\{2\}$};
        \draw  (1.5, 1.5) node{$\emptyset$};
        \draw  (4.5, 4.5) node{$\{ 1,2 \}$};
        
        \draw  (19/30*5, 0) node[below]{x};
        \draw  (7/10*5, 0) node[below]{y};
        \draw  (0,7/10*5) node[left]{y};
        \draw  (5/6*5, 0) node[below]{z};
        \draw (0,5) -- (5,5) -- (5,0);

\draw[dotted]  (0,7/10*5) -- (5,3.5);
        
        \draw[dotted] (19/30*5,0) -- (19/30*5,5);
        \draw[dotted] (7/10*5,0) -- (7/10*5,5);
        \draw[dotted] (5/6*5,0) -- (5/6*5,5);
    \draw[->] (2.5,2.5) node[left]{$\underline c(\theta_i)$}  -- (3.3,3.3);
    \draw[->] (0.5*5,4) node[left]{$\overline e(\theta_i)$}  -- (0.7*5-0.15,4);
    \draw[->] (5+0.2,2.4) node[right]{$\underline e(\theta_i)$}  -- (5/6*5+0.2, 7/10*5+0.2);
 
        \end{tikzpicture}
}}
\subfloat{ 
\label{subfig:M_pos2}
\scalebox{0.75}{
\begin{tikzpicture}[
declare function={
    Mc(\x)= ((900*\x*\x-361)/1920;
   Mc2(\x)= (700*\x*\x+269)/1920;
   Mc3(\x)= (1/480)*(40*(\x*\x)+161);
  }
  ]
\begin{axis}[%
            domain=0:1,
            yscale=1,
            axis lines=middle,
            xlabel={\large $\theta_i$},
            ylabel={\large $M_i$},
                       ymin=0,
            ymax=0.75,
            xmin=-0.1,
            xmax=1.1,
           samples=10,
             xlabel style={below right},
              xtick={0,19/30,7/10,5/6,1},
    xticklabels={0,$x$,$y$,$z$,1},
        ]
     \addplot [thick,domain=19/30:7/10-0.00001] {Mc(x)};
     \addplot [thick,domain=7/10:5/6-0.00001] {Mc2(x)};
     \addplot [thick,domain=5/6:1] {Mc3(x)};
          \draw [dotted] (80,50) -- (80,320);
    \end{axis}
\end{tikzpicture}
}}
 \caption{Optimal allocation (left) and optimal interim expected transfer $M_i(\theta_i)$ (right) for $\pi=\nicefrac{5}{8}$, and $\theta_i \sim U[0,1]$.}\label{fig:neg-again}
\end{center}
\end{figure}

In this example, $x=\nicefrac{19}{30}$, $y=\nicefrac{7}{10}$, and $z=\nicefrac{5}{6},$ but we can find $x<y<z$ in other settings with negative value network effects. The implementation has to be only slightly adjusted when $x<z<y$. We define
\eq{ \label{eq:xyzNeg} 
\begin{aligned}
&x: &\psi(x,2)+\psi(y,2)=c, \\
&y: &\psi(y,1)=c, \\
&z: &\psi(x,2)+\psi(\overline \theta,2)=c. \\
\end{aligned}
}
Again, we first describe the optimal allocation that our constructed game and equilibrium shall replicate.
Types $\theta_i \leq x$ never get the good and never pay anything. 
Types $\theta_i \in (x,y]$ never consume the good alone, 
and they consume the good jointly if type $\theta_j$ is large enough, but not too large,
\eq{ \label{eq:alloc-xy}
\begin{aligned}
\psi(\theta_i,2)+\psi(\theta_j,2) \geq \psi(\theta_j,1)
\iff \theta_j \leq \overline e(\theta_i), \mbox{ and } \\
\psi(\theta_i,2)+\psi(\theta_j,2) \geq c
\iff \theta_j \geq  \underline c(\theta_i).
\end{aligned}
}
\eq{ \label{eq:alloc-yzNO}
\mbox{Types $\theta_i \in (y,z]$ do not consume the good if } \theta_j> \overline e(\theta_i). 
}
Given $\theta_j \leq \overline e(\theta_i)$, $i$ consumes the good jointly with $j$ if 
$\psi(\theta_i,1) \leq \psi(\theta_i,2)+\psi(\theta_j,2)
\iff \theta_j \geq \underline e(\theta_i)$.  
\eq{ \label{eq:alloc-yzALONE}
\mbox{Types $\theta_i \in (y,z]$ consume the good alone if } \theta_j < \underline e(\theta_i),
} 
Finally, types $\theta_i \in [z,1]$ always get the good: alone if $\theta_j < \underline e(\theta_i)$, and jointly otherwise.

Now, we jointly construct the equilibrium and the game described in Figure \ref{subfig:implement_neg}.
 First, let actions be a pair $\alpha_i=(p_i,b_i) \in \mathbb R^2$. Next, define 
 \eq{ \label{eq:alloc-z1NEG}
 \overline p= \overline M_i(y), \quad \mbox{ and } \quad \overline{\overline{p}} = \overline M_i(z).
 }

\textbf{Strategy:}
Consider the following strategy $\sigma^*_i$:
\begin{enumerate}
    \item Types $\theta_i\in [0,x]$ select $(p_i,b_i)=(0,0).$
    \item Types $\theta_i\in (x,y]$ select $p_i=p^{x,y}(\theta_i)$ and $b_i= 0$, where 
    \ea{
    p^{x,y}: \quad p^{x,y}(\theta_i) = \overline M_i(\theta_i).}
\item Types $\theta_i\in (y,z]$ select $p_i=\overline p$ and $b_i= b^{x,y}(\theta_i)$, where 
    \ea{
    b^{y,z}: \quad \overline p + \overline Q_i^1(\theta_i)  b^{y,z}(\theta_i) = \overline M_i(\theta_i).}  
\item Types $\theta_i\in (z,1]$ select $p_i=\overline{\overline{p}}$ and $b_i= b^{z,1}(\theta_i)$, where 
    \ea{
    b^{z,1}: \quad \overline{\overline{p}} +  \overline Q_i^1(\theta_i)  b^{z,1}(\theta_i) = \overline M_i(\theta_i).}
    \end{enumerate}
Again, $p^{x,y}$, $b^{y,z}$ and $b^{z,1}$ are strictly increasing over their respective domains, because $\overline M_i$, $\frac{\overline M_i - \overline p}{\overline Q_i^1}$ and $\frac{\overline M_i -  \overline{\overline{p}}}{\overline Q_i^1}$ are strictly increasing.
    
    

\textbf{Outcome function:} Our goal is to construct an outcome function $g^*$  such that 
(i) it implements the optimal allocation in conjunction with a symmetric strategy profile $\sigma^*=(\sigma_1^*,\sigma_2^*)$, and (ii) $\sigma^*$ is a Bayesian Nash equilibrium. 

%

We first only consider actions consistent with $\sigma^*$, and we ensure that all types $\theta_i$ playing $\sigma_i^*$ get their designated expected outcome and payoff from the optimal mechanism $\langle \overline q,\overline m \rangle$ for all $\theta_j$ playing $\sigma^*_j$. 
First, we consider types $\theta_i\in (x,y]$.
For all $(p_i,b_i)$ with $p_i \in(0,\overline p]$ and $b_i=0$
set
\eq{ \label{eq:equi-xyNEG}
\begin{aligned}
   g^*_i \big( (p_i,b_i),(p_j,b_j) \big) = 
   \begin{cases}
			(1, p_i)  & \text{if }   b_j \leq \beta_1( p_i,p_j), \\
			(0,p_i) & \text{ otherwise},
		 \end{cases}\\
   \mbox{where }  
\beta_1( p_i, p_j) =\begin{cases}
    p_j -p^{x,y} ( \underline c ( (p^{x,y})^{-1} (p_i)))  & \quad \text{ if } \quad  p_j \in (0, \overline p)\\
    b^{y,z} ( \overline e ( (p^{x,y})^{-1} (p_i))) &\quad \text{ if } \quad  p_j=\overline p \\
    b^{z,1} ( \overline e ( (p^{x,y})^{-1} (p_i))) &\quad \text{ if } \quad  p_j= \overline{\overline{p}}
\end{cases}
\end{aligned}
}
such that given $\sigma^*$ and $\theta_i \in (x,y]$, $i$ can only get the good jointly, if and only if  $b_j \leq \beta_1( p_i, p_j)$.
Playing against types $\theta_j \in (x,y]$ who select actions $b_j=0$ and $p_j=p^{x,y}(\theta_j) \in (0,\overline p]$, the condition $b_j \leq \beta_1( p_i, p_j)$ implies
\begin{align*}
0 \leq  p^{x,y}(\theta_j) -p^{x,y} ( \underline c ( (p^{x,y})^{-1} (p_i))) \quad 
\iff \underline c(\theta_i) \leq \theta_j ,
\end{align*}
which is the second part of \eqref{eq:alloc-xy}, while the first part is implied because the relevant $p_j$ are only selected by types
$\theta_j \leq y = \overline e(x) \leq \overline e(\theta_i) \quad \forall \theta_i\geq x.$
For $p_j = \overline p$ (analogously for $p_j = \overline{\overline{p}} $), the condition $b_j \leq \beta_1( p_i, p_j)$ implies
\ea{
 b^{y,z}(\theta_j) \geq \beta_1 ( p^{x,y}(\theta_i), \overline p) \iff
 b^{y,z}(\theta_j) \geq  b^{y,z} ( \overline e ( (p^{x,y})^{-1} (p_i)))
\iff \theta_j \geq \overline e(\theta_i),
}
which is the first part of \eqref{eq:alloc-xy}. The second part also holds because prices $\overline p$ and $\overline{\overline{p}}$ are only selected by types $\theta_j\geq y = \underline c (x) \geq \underline c (\theta_i) \forall \theta_i\in(x, y].$

Next, we  consider types $\theta_i\in (y,z]$. For all $(p_i,b_i)$ with $p_i= \overline p$ and $b_i>0$, set
\eq{ \label{eq:equi-yzNEG}
\begin{aligned}
   g^*_i \big( (p_i,b_i),(p_j,b_j) \big) = 
   \begin{cases}
			(1, \overline p + b_i)  & \text{if }   b_i > \underline \beta_2(b_j,p_j) \\
   			(0, \overline p )  & \text{if }   b_j > \overline \beta_2(b_i,p_j)  \\
			(1,\overline p) & \text{ otherwise}.
		 \end{cases}
\end{aligned}
}
After defining $\overline \beta_2(b_i,\overline p)=b^{y,z}(\overline e((b^{y,z})^{-1}(b_i)))$, we have that $i$ does not consume the good in equilibrium if $\theta_j\in(y,z]$ and
\ea{
b_j=b^{y,z}(\theta_j) > b^{y,z}(\overline e((b^{y,z})^{-1}(b^{y,z}(\theta_i)))
\quad \iff \theta_j > \overline e(\theta_i),
}
which is \eqref{eq:alloc-yzNO}. Similarly, we ensure the same for $\theta_j\in(z,1]$ with 
$\overline \beta_2(b_i,\overline{\overline p})=b^{z,1}(\overline e((b^{y,z})^{-1}(b_i)))$. We can set $\overline \beta_2(b_i,p_j)=\infty$ for other $p_j$ because the associated types $\theta_j\in[x,y)$ shall never kick $i$ out of the audience.

After defining $\underline \beta_2(b_j,\overline p)= b^{y,z}(\underline e^{-1} (b^{y,z})^{-1}(b_j))$, we have that $i$ consumes the good alone and pays the exclusivity bid in equilibrium if $\theta_j \in (y,z]$ and
\ea{
b_i=b^{y,z}(\theta_i) > b^{y,z}(\underline e^{-1} (b^{y,z})^{-1}(b^{y,z}(\theta_j))
\quad \iff \underline e(\theta_i) > \theta_j,}
which is \eqref{eq:alloc-yzALONE}.
Similarly, we ensure the same for $\theta_j\in[z,1)$ with 
$\underline \beta_2(b_j,\overline{\overline p})= b^{y,z}(\underline e^{-1} (b^{z,1})^{-1}(b_j))$, and, for $\theta_j\in (x,y]$, we set $\underline \beta_2(0,p_j)= b^{y,z}(\underline e^{-1} (p^{x,y})^{-1}(p_j))$ for $p_j\leq \overline p$ with $b_j=0$.
In the remaining cases, $i$ and $j$ consume jointly.

Next, we  consider types $\theta_i\in [z,1)$. For all $(p_i,b_i)$ with $p_i= \overline{\overline{p}} $ and $b_i$ set
\eq{ \label{eq:equi-z1NEG}
\begin{aligned}
   g^*_i \big( (p_i,b_i),(p_j,b_j) \big) = 
   \begin{cases}
   			(1, \overline{\overline{p}} + b_i)  & \text{if }   b_i >\beta_3(b_j,p_j), \\
			(1, \overline{\overline{p}})  & \text{otherwise.}\\
   \end{cases}
\end{aligned}
}
 We define $\beta_3(b_j,\overline {\overline p})=\infty$ for all $b_j$ and $\beta_3(0,p)=0$ for all $p\leq \overline p$. In conjunction with the earlier construction steps of $g^*$, this implies that $i$ always consumes jointly with types $\theta_j> z$ without paying the exclusivity bid, and $i$ always pays the bid and consumes alone if $\theta_j\leq y$. To ensure that $i$ pays the exclusivity bid (and consumes alone) in the correct instances when $\theta_j \in (y,z]$, we define 
 $\beta_3(b_j,\overline p)=  b^{z,1}(\underline e^{-1} (b^{y,z})^{-1}(b_j))$
 such that, as desired,
 \ea{
 b_i = b^{z,1}(\theta_i) > b^{z,1}(\underline e^{-1} (b^{y,z})^{-1}(b^{y,z}(\theta_j))
 \quad \iff \underline e (\theta_i) > \theta_j.
 }





By construction, the interim expected payoff and allocation of all types given $\sigma^*$ in this game is the same as in the optimal mechanism. The reason that $\sigma^*$ is a Bayesian Nash equilibrium
is analogous to the last paragraph in the previous example after we make deviations outside of the range of $\sigma^*$ unprofitable in a similar fashion.

\bibliographystyle{elsarticle-harv}
{\small
\bibliography{twitch-bib}
}


\end{document}